\documentclass[english, 10pt]{article}

\usepackage{amsthm,amsmath, amssymb}
\usepackage{url} 
\usepackage[authoryear]{natbib}

\usepackage[colorlinks=true,linkcolor=blue,citecolor=blue,urlcolor=blue,breaklinks]{hyperref}
\usepackage{babel}
\usepackage{caption}
\usepackage{subcaption}
\usepackage{graphicx} 
\usepackage{color}
\usepackage{framed}
\usepackage{lscape}
\usepackage{rotating}
\usepackage{algorithm}
\usepackage[noend]{algpseudocode}
\usepackage{stefan_tex}
\usepackage{grffile}
\usepackage{multirow}

\newcommand{\ep}{\varepsilon}

\newcommand{\I}{\mathcal{I}}

\newtheorem{theorem}{Theorem}[section]

\newtheorem{proposition}[theorem]{Proposition}

\usepackage{geometry} 

\usepackage{etoolbox}
\DeclareMathOperator{\arccosh}{arcCosh}

\makeatletter
\def\BState{\State\hskip-\ALG@thistlm}
\makeatother
\makeatletter
\def\@footnotecolor{red}
\define@key{Hyp}{footnotecolor}{%
 \HyColor@HyperrefColor{#1}\@footnotecolor%
}
\def\@footnotemark{%
    \leavevmode
    \ifhmode\edef\@x@sf{\the\spacefactor}\nobreak\fi
    \stepcounter{Hfootnote}%
    \global\let\Hy@saved@currentHref\@currentHref
    \hyper@makecurrent{Hfootnote}%
    \global\let\Hy@footnote@currentHref\@currentHref
    \global\let\@currentHref\Hy@saved@currentHref
    \hyper@linkstart{footnote}{\Hy@footnote@currentHref}%
    \@makefnmark
    \hyper@linkend
    \ifhmode\spacefactor\@x@sf\fi
    \relax
  }%
\makeatother
\hypersetup{footnotecolor=blue}
\graphicspath{
 {../Monotone/Example 0 - Plot Objective/}
 {../Monotone/Example 1 - Basic Comparison - Monotone and Spjotvoll/}
 {../Monotone/Example 2 - Upper and Lower Bounds/}
 {../Monotone/Example 3 - Running Time/}
 {../Monotone/Example 4 - Objective and Derivatives/}
 {../Monotone/Example 5 - Subsample/Img/} 
 {../Monotone/Example 6 - Power Loss/Img/}
}

\begin{document}
\title{Weighted mining of massive collections of $p$-values by convex optimization}

\author{Edgar Dobriban\footnote{E-mail: \texttt{dobriban@stanford.edu}.}}
\date{Stanford University}

\maketitle

\begin{abstract}
Researchers in data-rich disciplines---think of computational genomics and observational cosmology---often wish to mine large bodies
of $p$-values looking for significant effects,
while controlling the false discovery rate or
family-wise error rate.
Increasingly, researchers also wish to
prioritize certain hypotheses, for example those thought to have
larger effect sizes, by upweighting, and to impose constraints
on the underlying mining, such as monotonicity along a certain
sequence.

We introduce \emph{Princessp},
a principled method for performing
weighted multiple testing by constrained convex optimization.
Our method elegantly allows one to prioritize certain hypotheses through upweighting
and to discount others through downweighting,
while constraining the underlying
weights involved in the mining process.
When the $p$-values derive from
monotone likelihood ratio families like the Gaussian
means model, the new method allows exact solution
of an important optimal weighting problem
previously thought to be nonconvex and computationally infeasible.
Our method scales to massive dataset sizes.

We illustrate the applications of Princessp
on a series of standard genomics datasets and
offer comparisons with several previous `standard' methods. 
Princessp offers both ease of operation and
the ability to scale to extremely large problem sizes. 
The method is available as open-source software from 
\url{github.com/dobriban/pvalue_weighting_matlab}.


\end{abstract}


\section{Introduction}

\subsection{Large-scale inference}

Large-scale inference \citep{efron2012large} is the new term of art for an
emerging set of activities that crosscut all of scientific data analysis.
Owing to new data collection and computing technology, massive inference problems are appearing that
challenge traditional statistical approaches.

An emblem of the new era is the compilation of vast
collections of $p$-values to be mined and analyzed.

\begin{itemize}
\item 
Workers in computational biology, working with different forms of genome-derived data---gene expression,  Single Nucleotide Polymorphism (SNP), Copy Number Variation (CNV), etc.---routinely compute a $p$-value for each relevant genomic locus, providing
statistical evidence that the locus is associated with some outcome. When done for hundreds of thousands or even millions of
such sites, massive collections of $p$-values emerge \cite[see e.g.,][for representative examples]{tusher2001significance, burton2007genome}.

\item  Workers in extragalactic astronomy look for evidence of non-Gaussianity in the
cosmic microwave background (CMB), by compiling $p$-values for tests of non-Gaussianity
associated with different regions, orientations and scales of the CMB data.

\end{itemize}

The emergence of large collections of $p$-values is happening in field after field.
In fact scientific publication itself can now be viewed as a collection of numerical $p$-values:

\begin{itemize}
\item Researchers such as Ioannidis, Leek, and collaborators have conducted extensive text scraping exercises,
extracting  millions of $p$-values out of the
abstracts and bodies of published scientific papers \citep[e.g.][]{allen2008systematic,jager2013estimate}.
\end{itemize}

This emergence creates an urgent demand for tools
that can conveniently and reliably facilitate the main tasks of large-scale inference.

\subsection{Weighted testing of individual hypotheses}

A primary task in large-scale inference (LSI) is simultaneous testing of {\it individual null hypotheses}---the focus of this paper.
Intuitively, this task allows one to discover the `signals'  amid a sea of `noise' $p$-values, while providing some form of rigorous error control.

Many popular simultaneous significance tests are available, such as the Bonferoni method, Holm's method \citep{holm1979simple}, or the Benjamini-Hochberg method \citep{benjamini1995controlling},
but all suffer from a fundamental limitation: they treat the hypotheses as equal, and thus are unable to exploit important prior information about the likely significance of some hypotheses.
Incorporating prior knowledge into the multiple testing problem holds the promise of achieving higher power. 

Weighted Bonferroni testing of $p$-values is an increasingly popular method to prioritize certain effects over others.
We assign a weight $w_i$ to the $i$-th $p$-value, and reject the $i$-th null hypothesis if
$P_i/w_i \leq \alpha/\sum_i w_i$. The larger the weight, the less stringent the rejection threshold, and the easier it is to reject a null. See the review articles of \cite{roeder2009genome} and \cite {gui2012weighted}.

There is a growing trend to use weighted testing of $p$-values, especially in the biomedical literature,
where a number of successful applications to genome-wide association studies (GWAS) are listed in Table \ref{recent_apps}.

\begin{table}[h]
\centering
\caption{Some recent uses of prior information in biomedical studies}
\label{recent_apps}
\resizebox{0.9\columnwidth}{!}{
\begin{tabular}{|l|l|l|}
\hline
Source                                            & Current GWAS                 & Prior \\ \hline
\cite{saccone2007cholinergic}                  & nicotine dependence            & nicotine receptor status     \\ \hline
\cite{andreassen2013improved}                   & schizophrenia            & cardiovascular disease  GWAS      \\ \hline
\cite{li2013using} & childhood asthma    & eQTL \\ \hline
\cite{rietveld2014common} & cognitive performance    &educational attainment GWAS \\ \hline
\cite{fortney2015genome}  &extreme longevity        &age-related traits     \\ \hline
\cite{sveinbjornsson2016weighting}              & aggregate over broad set & annotation in established hits        \\ \hline
\end{tabular}
}
\end{table}

Unfortunately, the existing $p$-value weighting procedures have limited
ability to cope with important practical requirements.


\subsection{Optimal weighting and  obstacles to practitioner acceptance}

A motivational example arises when we can assign---say, based on
earlier experimental measurements---
 prior effect sizes $\mu_i$ to the different hypotheses. For instance, in biomedical research we may have access to data from related studies on the same trait. Our goal is to
perform weighted testing in a fashion that rigorously controls the
type I error, but offers optimal detection power,
under the alternative hypothesis that the specified effect sizes are true.
Equivalently, we want to apply weighted multiple testing of $p$-values with the
so-called Spj\o tvoll weights. See \cite{Spjotvoll1972optimality}, \cite{westfall1998using},  \cite{genovese2006false}, \cite{rubin2006method}, etc.

The Spj\o tvoll weights can have certain unusual properties which render them
unacceptable to practitioners. Among these are the fact that the weights themselves
need not be monotone in the effect sizes $\mu_i$; and the fact that some of the formally optimal weights can be
very close to zero, causing the weighted $p$-value statistic to behave in an unstable fashion.
Practitioners simply don't consider it safe and effective to rely on weights with such properties.
Existing approaches developed to cope with these unfortunate aspects of the optimal weights
involve nonconvex optimization \citep{westfall1998using,westfall2001optimally},
and so don't scale well to large problem sizes. So at the moment, there is no computationally effective way
to impose `practitioner' constraints such as monotonicity and strict positivity;
in the era of practical large scale of inference we cannot yet apply optimal weighting
to massive collections of $p$-values.



\subsection{Princessp: New approach based on convex optimization}

This paper develops \emph{Princessp}, a new approach to weighted Bonferroni multiple testing of $p$-values
which resolves the practitioner obstacles just mentioned. Princessp employs convex optimization and so
can scale to very large problem sizes. It allows linear inequality constraints to be imposed, constraints
which can include monotonicity of weights in the effect size and also lower bounds keeping all weights at 
a definite distance from zero. Hence the key priorities of practitioners---scale, monotonicity and
definite positiveness of the weights---are all handled by our new Princessp method.

In more detail, we now enumerate our main contributions: 


\begin{enumerate}

\item We formulate Princessp as the constrained convex optimization problem of computing $p$-value weights maximizing average power of the weighted Bonferroni method (Sec. \ref{cvx}). The concavity of the receiver-operating characteristic (ROC) is the key property enabling this approach. We establish concavity in the important setting of of two-sided tests in exponential families. 

\item We give several examples of practical convex constraints to be used with Princessp (Sec. \ref{examples}). Our framework allows us to consider in a unified way several types of weights (such as stratified or smooth weights) that have been considered in previous work, and also leads to new methods (such as bounded and monotone weights, Sec. \ref{mon}). 

\item We develop an efficient interior-point method for monotone bounded weights under normal observations (Sec. \ref{sec:optim}). The direct use of interior-point methods leads to a severely ill-conditioned KKT system, and is inapplicable to large-scale problems. Therefore, we develop a new subsampling algorithm to avoid the ill-conditioning (Sec. \ref{refinements}). We establish that subsampling is fast and accurate.

This allows us to efficiently solve large-scale problems with potentially tens of millions of hypotheses, such as those common in genomics. 
A MATLAB implementation of this method, as well as the code to reproduce our computational results is available at \url{github.com/dobriban/pvalue_weighting_matlab}, and will be provided in the software supplement. 

\item We illustrate Princessp on a series of standard genome-wide association study (GWAS) datasets (Sec. \ref{sec:data}). The boundedness of the weights seems to be important here, and leads to a good empirical performance compared to the state of the art weighted Bonferroni methods.
\end{enumerate}


We start in Sec.\ \ref{cvx} by explaining Princessp first in the Gaussian case, then more generally for monotone likelihood ratio families and two-sided tests. We give several concrete examples of convex constraints (Sec. \ref{examples}). We study the important special case of monotone $p$-value weights (Sec. \ref{mon}), illustrating their empirical performance on standard GWAS datasets and simulated data (Sec. \ref{sec:data}). Finally, we develop numerical methods to compute them (Sec. \ref{sec:mon_w}).

\section{Princessp: Convex $p$-value weights}
\label{cvx}

To explain our framework, we start with the normal means model. Suppose we observe test statistics $T_i \sim \mathcal{N}(\mu_i,1)$, $i=1,\ldots ,J$ and want to test the individual hypotheses $H_i:\mu_i \ge 0$ against $\mu_i <0$ simultaneously. Here $\mu_i$ are the standardized effect sizes, which are arbitrary fixed parameters.

We construct $p$-values $P_i = \Phi(T_i)$, and reject the $i$-th null $H_i$ if $P_i \le q w_i$, with $w_i\ge0$ and $\sum w_i = J$. This is the weighted Bonferroni procedure, which controls the family-wise error rate---the probability of making any false rejections---at level $\alpha = Jq$. The weights are based on independent prior data. Without loss of generality we assume $w_i\le 1/q$.

To optimize the choice of weights, it is customary to proceed in two stages. First we assume that we know $\mu_i$, and derive formally optimal ``oracle'' weights $w_i = w(\mu_i)$ depending on the unknown parameters. Second, when $\mu_i$ are unknown but we can construct estimators $\hat \mu_i$ based on some empirical data, we use the weights $w(\hat \mu_i)$. In this paper, we will focus on the first stage. In practice this corresponds to oracle weighting conditioning on $\hat\mu_i$. Thus we do weighting as if the true parameters were $\hat \mu_i$, but we do not use those estimates for other statistical inference tasks. We return to this issue at the end of the section. 

If $\mu_i$ are known, we can assign zero weight to any non-negative $\mu_i$; so in the remainder we will assume that all $\mu_i<0$. We can choose the remaining weights to maximize the expected number of discoveries. The statistical power to reject $H_i$ when the true parameter is $\mu_i$ equals $\Phi\left(\Phi^{-1}\left(qw_i\right)-\mu_i\right)$, leading to the optimization problem 
\begin{equation}
\label{Spjotvoll_weights}
 \max_{w\in [0,1/q]^J} \mbox{  } \sum_{i=1}^{J} \Phi\left(\Phi^{-1}\left(qw_i\right)-\mu_i\right) \,\,\, \mbox{   s.t. }  \sum_{i=1}^{J} w_i = J.
\end{equation}

This class of weighted multiple testing problems was studied first by \cite{Spjotvoll1972optimality}, thus we call this the Gaussian \emph{Sp\o tvoll weights} problem.
\cite{rubin2006method} and  \cite{roeder2009genome} showed that the weights are $w_i = \Phi(\mu_i/2 + c/\mu_i)/q$, for the unique $c>0$ such that the weights sum to $J$. They studied the problem parametrized by the critical value $c_i = \Phi^{-1}(qw_i)$ instead of the weight $w_i$. In that parametrization, the objective is not concave, and therefore it is unclear how to efficiently incorporate additional constraints. As a key observation for this paper, in the weight parametrization, the function $w \to \Phi\left(\Phi^{-1}\left(w\right)-\mu\right)$ is the ROC curve of the Gaussian test, and is strictly concave if $\mu<0$ \citep[e.g.,][p. 101]{lehmann2005testing}.

Our Princessp method takes this observation to its natural conclusion, allowing arbitrary convex constraints for the weights. Specifically, let $f_k: [0,1/q]^J\to \mathbb{R}$, $k=1, \ldots, K$ be convex and twice continuously differentiable functions, encoding convex inequality constraints $f_k(w) \le 0$. Let $A$ be an $L \times J$ matrix encoding equality constraints $Aw =b$. We let the $1 \times J$ vector $(1,1,\ldots,1)$ be the first row of $A$, and $J$ be the first entry of $b$. 
We assume $\rank A <J$. 
In the following examples, one can easily check these conditions.

The Princessp problem of computing constrained optimal weights for Bonferroni multiple testing is thus, in the Gaussian case: 
\begin{align*}
 \max_{w\in [0,1/q]^J} \mbox{  } & \sum_{i=1}^{J} \Phi\left(\Phi^{-1}\left(qw_i\right)-\mu_i\right) \,\,\, \\
   \mbox{   s.t. }  &f_k(w) \le 0, \,\,\,\,\,\, k = 1,\ldots,K \\
   &Aw=b. 
\end{align*}

More generally, beyond one-sided Gaussian tests, we replace $\Phi\left(\Phi^{-1}\left(qw_i\right)-\mu_i\right)$ with the appropriate ROC curves. For instance, if we want to perform two-sided tests of a null hypothesis $\mu_i=0$ against $\mu_i \neq 0$, then the analogous objective function is
\begin{equation*}
 \sum_{i=1}^{J} \left [\Phi\left(\Phi^{-1}\left(qw_i/2\right)-\mu_i\right) + \Phi\left(\Phi^{-1}\left(qw_i/2\right)+\mu_i\right) \right].
\end{equation*}
We will show in Proposition \ref{two_exp} that the objective is concave, despite only one term in each pair being concave in $w_i$. That result applies more generally to continuous exponential families. For one-sided tests, we can handle even more general observations, arising from monotone likelihood ratio families (Sec. \ref{mlr}).

One of the most important convex constraints is \emph{monotonicity}. For Gaussian data, suppose that our guesses for the means are sorted such that $0 < |\mu_1| \le |\mu_2| \le  \ldots \le  |\mu_J|$. Monotonicity requires that larger absolute effects $|\mu_i|$ have a larger weight, so that $w_1 \le w_2 \le  \ldots \le  w_J $. Unconstrained Spj\o tvoll weights are \emph{not} monotone in general. However, monotonicity is clearly desirable in practice. We discuss monotone weights in Sec. \ref{mon}.

In Sec. \ref{sec:optim}, we will also explain how to solve the Gaussian problem numerically. Despite its simplicity, the problem has interesting numerical aspects. We show that the log-barrier interior point method converges \emph{in theory}, but it has serious numerical difficulties; namely, a severely ill-conditioned KKT system. For monotone weights, we show how to avoid ill-conditioning using a subsampling approach.

Returning to the issue mentioned above, in practice we will have some empirical estimates $\hat\mu_i$ of the effect sizes, and solve the Princessp problem with those values substituted for $\mu_i$, effectively conditioning on $\hat\mu_i$. Thus, in practice we will use the optimization formulation merely as a means to efficiently specify an algorithm for use on real data. The formal optimality properties of those plug-in weights are beyond our present scope.

\subsection{Monotone likelihood ratio}
\label{mlr}

For one-sided tests we have convex problems very generally in monotone likelihood ratio (MLR) families. Suppose $f_{\theta}(x)$ is a real-parameter family of densities with respect to Lebesgue measure, supported on a common open interval $\I=(a,b)$, where the endpoints can be infinite. We assume that $f_{\theta}(x)>0$ for all $x \in \I$, so that the distribution functions $F_{\theta}(x) = \int_{a}^{x} f_{\theta}(y) dy$ are differentiable and strictly increasing on $\I$. Then the inverse functions $F_{\theta}^{-1}(x)$ are well-defined for $x\in(0,1)$.  Our key assumption is that $f_{\theta}(x)/f_{\theta'}(x)$ has  a monotone increasing likelihood ratio in $x$, for $\theta>\theta'$.

Define the ROC curve $G(x):=G_{\theta,\theta'}(x) = F_{\theta}[F_{\theta'}^{-1}(x)]$. This gives the power at $\theta$ of a one-sided test with rejection region $X \le c$, as a function of the level $x$ under $f_{\theta'}$.

\begin{proposition}[MLR and concavity, e.g., \cite{lehmann2005testing} p. 101]
\label{mlr_ccv} Suppose $f_{\theta}(x)$ is a family of densities with monotone increasing likelihood ratio in $x$, and satisfies the assumptions above. Then the ROC curve $G_{\theta,\theta'}(x)$ is strictly convex in $x$ for $\theta>\theta'$, and strictly concave for $\theta<\theta'$.

\end{proposition}

\begin{proof}
Part of this result is stated, but not proved, on  p. 101 of \cite{lehmann2005testing}; therefore we provide a proof. We have $G'(x) = f_{\theta}[F_{\theta'}^{-1}(x)]/f_{\theta'}[F_{\theta'}^{-1}(x)]$. Hence, $G'$ is strictly increasing for $\theta>\theta'$, and strictly decreasing if $\theta<\theta'$. Therefore, $G$ is strictly convex for $\theta>\theta'$, and strictly concave for $\theta<\theta'$.
\end{proof}

Suppose now that we have test statistics $T_i \sim f_{\theta_i}$ and we are testing $\theta_i \ge \theta_{0i}$ against $\theta_i < \theta_{0i}$ simultaneously. As is well known, the uniformly most powerful individual test rejects for $T_i < c_i$. With the $p$-values $P_i = F_{\theta_{0i}}(T_i)$, the weighted Bonferroni test rejects the $i$-th null if $P_i \le q w_i$. 

As before, suppose we have prior guesses $\theta_{1i} < \theta_{0i}$ for the parameters. The power of the $i$-th test at  $\theta_{1i}$ is $G_{\theta_{1i},\theta_{0i}}(q w_i) = F_{\theta_{1i}}[F_{\theta_{0i}}^{-1}(q w_i)]$. To optimize the average power of the weighted Bonferroni method for this set of alternatives, we must solve
\begin{equation}
\label{mlr_weights}
 \max_{w\in [0,1/q]^J} \mbox{  } \sum_{i=1}^{J} G_{\theta_{1i},\theta_{0i}}(qw_i) \,\,\, \mbox{   s.t. }  \sum_{i=1}^{J} w_i = J.
\end{equation}
By Proposition \ref{mlr_ccv}, this objective is concave, and can be solved efficiently. In this setup, the Princessp problem adds general convex problem. This extension of Princessp to MLR allows us to handle one-sided tests in the following well-known statistical models: 
\begin{itemize}
\item One-dimensional exponential families with continuous support: Clearly, exponential families $p_{\theta}(x) = \exp(\theta \cdot x - A(\theta))$ have monotone likelihood ratio in $x$. Our framework includes one-sided tests on the natural parameter, such as the mean of a normal family $\mathcal{N}(\mu,\sigma^2)$ ($\sigma^2$ fixed), the variance in a normal family $\mathcal{N}(\mu,\sigma^2)$ ($\mu$ fixed), beta distributions (with one shape parameter fixed), or gamma distributions (with either the shape or the rate parameter fixed). 

\item Non-central $t,F,\chi$ distributions: The distribution of the $t$-statistic with fixed degrees of freedom and non-centrality parameter $\mu/\sigma$ has MLR in $t$ (e.g., \cite{lehmann2005testing}, p. 224). Similarly, the distribution of the non-central $F$ and $\chi^2$-statistics with fixed degrees of freedom have MLR (e.g., \cite{lehmann2005testing}, p. 307).

\item Location families: Location families $f_{\theta}(x) = f(x-\theta)$ with strictly positive and almost surely smooth log-concave density $f$ have monotone likelihood ratio. Examples include double-exponential (Laplace) distributions with $f_{\theta}(x) = \exp(-|x-\theta|)/2$ and logistic distributions $f_{\theta}(x) = \exp(x-\theta)/(1+\exp(x-\theta))^2$. 
\end{itemize}

We emphasize that the assumptions on the family of densities $f_{\theta}(x)$ made in this section are essential, and not merely technical. Indeed, if the distributions $F_{\theta}$ do not have a density, then the ROC curve may be discontinuous, rendering the optimization problem discontinuous, and definitely nonconvex. For instance, if $F_{p}$ denotes the distribution function of a Bernoulli variable with success probability $p$, then the ROC is the step function
$$F_q(F_p^{-1}(x)) = 
\left\{
	\begin{array}{ll}
		1  & \mbox{\, if \, }  x>p, \\
		q & \mbox{\, if \, } x \in (0,p],\\
		0 & \mbox{\, if \, } x \le 0.
	\end{array}
\right.
$$ 
Similar problems occur when the densities can be zero, or when the supports are non-overlapping.
%

An even more general convexity property holds for likelihood-ratio tests. As is well known in statistical folklore, and formalized for instance in Proposition 3.1 of \cite{pena2011power}, the ROC curve of LR tests is often automatically concave.  However, in our applications, the alternative is usually composite, parametrized by a scalar $\theta \in H_1$, not just a simple alternative. 
There is often some uncertainty in specifying $\theta$ from prior data. For this reason it is important to have good uniform optimality properties against composite alternatives. For this one essentially needs the MLR, which is why we chose that setting here.

\subsection{Two-sided tests in exponential families}

For two-sided tests, we can guarantee that Princessp works in exponential families, a setting more restricted than MLR, but more general than Gaussian. Suppose that we have observations $T_i \sim p_{\theta_i}(x)$ with density $p_{\theta}(x) = \exp(\theta \cdot x - A(\theta))$ with respect to some dominating measure. We want to test $\theta_i = \theta_{0i}$ against $\theta_i \neq \theta_{0i}$. Suppose we reject the $i$-th null if $|T_i - \theta_{0i}|>c_i$ for critical values $c_i$ such that the test has level $\alpha$. Let $G(\alpha):=G_{\theta_{1i},\theta_{0i}}(\alpha)$ be the power of this test at an alternative hypothesis $\theta_{1i}$. We analyze the convexity of the ROC curve in this setting. We are also interested in strong convexity, because it is needed for the convergence of the log-barrier method.

\begin{proposition}[Concavity for two-sided tests in exponential families] Consider two-sided tests $\theta = \theta_{0}$ against $\theta \neq \theta_{0}$ for observations following the exponential family $p_{\theta}$. 

 \begin{enumerate}
\item If $|\theta_1|>|\theta_0|$, the ROC curve $G_{\theta_{1},\theta_{0}}(\alpha)$ at $\theta_1$ is a concave function of the level $\alpha$. If $|\theta_1|<|\theta_0|$, $G_{\theta_{1},\theta_{0}}(\alpha)$ is convex. 
\item Moreover, if $|\theta_1|>2|\theta_0|$, then $G$ is strongly concave on any interval $\alpha \in (0,1-\ep)$, $\ep>0$. Similarly, if $2|\theta_1|<|\theta_0|$, $G$ is strongly convex on any interval $\alpha \in (0,1-\ep)$, $\ep>0$.
\end{enumerate}
\label{two_exp}
\end{proposition}

\begin{proof}
Let $F_{\theta}$ be the cumulative distribution function of the observations under parameter $\theta$. The critical value $c$ for the test is determined by the equation $1-[F_{\theta_0}(\theta_0+c)-F_{\theta_0}(\theta_0-c)]=\alpha$. Our assumptions guarantee that $c$ is well-defined. Indeed, $a(c)=[F_{\theta_0}(\theta_0+c)-F_{\theta_0}(\theta_0-c)]$ takes value $a(0)=0$, is differentiable with non-negative derivative $a'(c) = f_{\theta_0}(\theta_0+c)+f_{\theta_0}(\theta_0-c)\ge0$, and moreover, has limit $1$ as $c\to\infty$. In addition, $a'(c)>0$ for all $c$ such that either $\theta_0+c$ or $\theta_0-c$ belong to the support of $F_{\theta_0}$. Therefore, if $\theta_0 \in \I$, then $a'(c)>0$ precisely on an interval $[0,c_0)$ for some $c_0$. This implies that for any $\alpha \in (0,1)$, the equation $a(c) = 1-\alpha$ has a unique solution on $[0,c_0)$, showing that $c$ is well-defined. 

Moreover, $dc/d\alpha = -1/[f_{\theta_0}(\theta_0+c)+f_{\theta_0}(\theta_0-c)]$ for $c \in [0,c_0)$. In addition, the power of the test equals $p = 1-[F_{\theta_1}(\theta_0+c)-F_{\theta_1}(\theta_0-c)]$, therefore $dp/d\alpha = - [f_{\theta_1}(\theta_0+c)+f_{\theta_1}(\theta_0-c)]\cdot dc/d\alpha$. We will only consider $|\theta_1|>|\theta_0|$, the other case is analogous.

\begin{enumerate}
\item 
To show $p$ is concave for $|\theta_1|>|\theta_0|$, we must prove that $dp/d\alpha$ is decreasing in $\alpha$. Since $\alpha$ can be viewed as a decreasing function of $c \in [0,c_0)$, it is enough to show that $g:=dp/d\alpha$ is increasing in $c \in [0,c_0)$, where 
$$g(c) = \frac{f_{\theta_1}(\theta_0+c)+f_{\theta_1}(\theta_0-c)}
{f_{\theta_0}(\theta_0+c)+f_{\theta_0}(\theta_0-c)}.$$
Since we are in an exponential family, we see 
$$g(c) = \exp[A(\theta_0)-A(\theta_1)-\theta_0(\theta_0-\theta_1)]\cdot\frac{\exp(\theta_1c)+\exp(-\theta_1c)}
{\exp(\theta_0c)+\exp(-\theta_0c)}.$$
It is a simple calculus exercise to show that this function is increasing in $c\in[0,c_0)$--- and in fact on $[0,\infty)$---if $|\theta_1|>|\theta_0|$. 

\item To show strong concavity, we must check in addition that $d^2p/d\alpha^2<c_0<0$ for a constant $c_0$. Note that $d^2p/d\alpha^2 = dg/d\alpha = dg/dc\cdot dc/d\alpha$. Denoting by $K>0$ constants that do not depend on $c$ (or $\alpha$) and whose meaning may change from line to line, we get 
$$\frac{dg}{d\alpha} = -K \frac{\theta_1 \sinh(\theta_1c)\cosh(\theta_0c) - \theta_0 \sinh(\theta_0c)\cosh(\theta_1c)}{\cosh(\theta_0c)^3}. $$
Without loss of generality we may assume that $\theta_1>\theta_0\ge0$. Then the numerator equals 
$(\theta_1-\theta_0)\sinh(\theta_1c)\cosh(\theta_0c) 
+ \theta_0 [\sinh(\theta_1c)\cosh(\theta_0c)  - \sinh(\theta_0c)\cosh(\theta_1c)]
\ge (\theta_1-\theta_0)\sinh(\theta_1c)\cosh(\theta_0c)$. 
Hence, 
$$\frac{dg}{d\alpha} \le -K \frac{\sinh(\theta_1c)}{\cosh(\theta_0c)^2}. $$
Since $\alpha<1-\ep$, we have $c>c_{\ep}>0$. Clearly, for small $c>c_{\ep}>0$, $\sinh(\theta_1c),\cosh(\theta_0c)$ are bounded away from $0$ and $\infty$, so that $\frac{dg}{d\alpha}<c_0<0$ there. For large $c$, $\sinh(c) \sim \cosh(c) \sim \exp(c)$. Therefore, $\sinh(\theta_1c)/\cosh(\theta_0c)^2 \sim \exp([\theta_1-2\theta_2]c)$. So, if $\theta_1 > 2\theta_2$, then $dg/d\alpha <c_0 <0$ for all $c$, showing that the power function is strongly concave.
\end{enumerate}
\end{proof}

An example of special interest is testing normal observations $T_i\sim\mathcal{N}(\mu_i,1)$  for $\mu_i=0$ against $\mu_i \neq 0$. Then the two-sided analogue of Spj\o tvoll weights amounts to the optimization problem
\begin{equation}
\label{two_sided_Spjotvoll_weights}
 \max_{w\in [0,1/(2q)]^J} \mbox{  } \sum_{i=1}^{J} \Phi\left(\Phi^{-1}\left(qw_i/2\right)-\mu_i\right) + \Phi\left(\Phi^{-1}\left(qw_i/2\right)+\mu_i\right) \,\,\, \mbox{   s.t. }  \sum_{i=1}^{J} w_i = J.
\end{equation}

Note that the first term is concave in $w_i$, but the second term is not. However, their sum is concave by Proposition \ref{two_exp}.  This reinforces that the results in this section go beyond those from the previous section.

We can also find a nearly explicit form for the two-sided optimal Gaussian weights. This leads to a fast algorithm for computing them, as well as insights about their behavior (see Sec. \ref{mon}).

\begin{proposition}[Explicit optimal weights for two-sided Gaussian tests] Define the function $w(\mu;\lambda) = 2\Phi\left(  - \arccosh\left[\lambda \exp(\mu^2/2)/q \right]/|\mu| \right)/q$ and the constant $m = \min_i \{\mu_i^2/2\}$. Define also the function $H(\lambda) =\sum_{i} w(\mu_i;\lambda)$. If all $\mu_i \neq 0$, and if $H(q\exp(-m)) \ge J$, then the optimal weights for the two-sided normal problem are $w_i = w(\mu_i;\lambda)$, where $\lambda\ge q\exp(-m)$ is the unique constant such that $\sum_i w_i(\lambda) = J$. 
\label{two_side_formula}
\end{proposition}

\begin{proof}
Define $g(w,\mu) = \Phi\left(\Phi^{-1}\left(qw/2\right)-\mu\right) + \Phi\left(\Phi^{-1}\left(qw/2\right)+\mu\right) $. The Lagrangian for the two-sided normal problem is 
$$L(w_1,\ldots,w_J; \lambda) = \sum_{i=1}^J [ g(w_i,\mu_i) - \lambda w_i].$$
Now,  $\partial [g(w,\mu) - \lambda w]/\partial w = q\exp(-\mu^2/2) \cosh(\mu z)-\lambda$, where $z = \Phi^{-1}(qw/2)$. Therefore, we have $ \partial  L(w_1,\ldots,w_J; \lambda) /\partial w_i = 0$ for all $i$ iff, denoting  $z_i = \Phi^{-1}(qw_i/2)$
$$\lambda = q\exp(-\mu_i^2/2) \cosh(\mu_i z_i),$$
or equivalently 
$\lambda \exp(\mu_i^2/2)/q = \cosh (\mu_i z_i)$. This is equivalent to $\arccosh\left[\lambda \exp(\mu_i^2/2)/q\right] = |\mu_i| |z_i|$. Since $z_i \le 0$, from this we obtain that the optimal $w_i$ maximizing the Lagrangian $L$ have the desired form $w_i = 2\Phi\left(- \arccosh\left[\lambda \exp(\mu_i^2/2)/q \right]/|\mu_i| \right)/q$ for fixed $\lambda$, if $\lambda \ge q \exp(-\mu_i^2/2)$. If we can show that there exists a $\lambda$ such that the constraint $\sum w_i = J$ holds, then it will follow that those weights also solve the original constrained problem. 

For this, note that the function $\lambda \to \arccosh\left[\lambda \exp(\mu_i^2/2)/q \right]$ is continuous and strictly increasing on $[q \exp(-\mu_i^2/2),\infty)$, going from $0$ to $\infty$. Therefore, the weight $w_i(\lambda)$ is a strictly decreasing function of $\lambda$ on $[q \exp(-\mu_i^2/2),\infty)$, going from $2\Phi(0)/q = 1/q$ to $0$. Hence $H(\lambda) = \sum_i w_i(\lambda)$ is strictly decreasing on $[q\exp(-m),\infty)$, with limit 0 at $+\infty$. By the assumption $H(q\exp(-m))\ge J$, it follows that there is a unique $\lambda \ge q\exp(-m)$ such that the constraint holds. This finishes the proof. 
\end{proof}

\subsection{Examples of convex constraints}
\label{examples}

The key flexibility of Princessp lies in the wide variety of user-specified convex constraints that it can be used with. To help potential users see how this is relevant in a variety of contexts, we now give several examples of such constraints. 

\begin{enumerate}

\item{\bf Monotone weights.} Suppose in the Gaussian case that the means are sorted such that $0 < |\mu_1| \le |\mu_2| \le  \ldots \le  |\mu_J|$. Monotone weights require that larger absolute effects $|\mu_i|$ have a larger weight, so that $w_1 \le w_2 \le  \ldots \le  w_J $. 

Unconstrained Spj\o tvoll weights are not monotone in general. However, monotonicity is intuitively desirable, as \emph{larger effects are worth more}. Many popular weighting schemes are monotone, including exponential weights $w_i  \propto \exp(\beta|\mu_i|)$ and cumulative weights $w_i  \propto \Phi(|\mu_i|-B)$, for $B>0$, both proposed in \cite{roeder2006using}. The weights of \cite{li2013using} based on normalized versions of $(-2\log \tilde P_j)^{1/2}$, where $\tilde P_j$ are $p$-values from independent expression Quantitative Trait Locus (eQTL) studies, are also monotone in the strength of the prior information. Princessp offers a principled way to incorporate monotonicity as a constraint in any weighting scheme.

\item{\bf Bounded weights.} Boundedness requires that $l \le w_i \le u$ for two constants $0 \le l < 1 < u$. This ensures that the current $p$-values $P_i$ get multiplied by at most $1/l$. Hence, if $l=1/2$ and the original $p$-value cutoff is---say---the conventional threshold $q = 5\cdot 10^{-8}$ for genome-wide significance, then each hypothesis with $p$-value $P_i \le ql = 2.5 \cdot 10^{-8}$ is significant, \emph{regardless of the strength of prior information}. This is desirable as the prior information is often unreliable. Small weights are risky, as they can weaken strong $p$-values. For instance if $P_i = 10^{-10}$ but $w_i = 10^{-9}$, then the weighted $p$-value is a meager $Q_i = 0.1$. 

Choosing the lower bound presents a bias-variance tradeoff. A large lower bound, say 0.5, has ``low variance'' as the weighted $p$-values do not change much. However it potentially has ``high bias'' as the truly optimal weights may be small. Empirically, we found that a lower bound in the range 0.01-0.5 often works well (see the data analysis section).

Using an upper bound $u$, no hypothesis with $p$-value $P_i > qu$ is rejected. For instance, if $q = 5\cdot 10^{-8}$ and $u=10$, then only the $p$-values $P_i \le 5\cdot 10^{-7}$ have a chance of being rejected. The upper bound ensures that we do not place too much weight on any hypothesis. Princessp allows the user to conveniently specify arbitrary bounds for the weights. 

\item{\bf Stratified weights.} Let $S_1, \ldots, S_k$ be disjoint subsets of $\{1,\ldots,J\}$.  With stratified weights, we want to assign an equal weight to the hypotheses in the same subset $S_i$, for each $i$. This is the linear equality constraint that $w_j = w_k$, for all $j,k \in S_i$, and for all $i$. 

Stratification is natural when effects are grouped, say into different functional classes of genetic variants in genome-wide association studies. The stratified FDR method \citep{sun2006stratified}, is closely related.

Related ideas were also used more informally in previous work. In a study of nicotine dependence with approximately 3700 genetic variants, \cite{saccone2007cholinergic} improve power by giving nicotine receptor genes ten times the weight of other candidate genes. More generally, \cite{roeder2009genome} study binary weights that take only two possible values.
Similarly, \cite{sveinbjornsson2016weighting} assign weights to classes of variants in GWAS based on functional annotation. As another example, \cite{eskin2008increasing} groups polymorphisms by location in the genome, assigning them to a tag and requiring the same weight for each tag. Stratified weights can be incorporated in a direct was as convex constraints in Princessp. 

\item{\bf Smooth weights.} If the hypotheses have some spatial structure, then it is reasonable to require that the weights be ``smooth'' with respect to this structure. For instance, genetic variants are aligned on chromosomes, and one could require the weights to be smoothly varying as a function the position.  This can be achieved in many ways, borrowing from the vast literature on regularization in statistics. A simple example is to add a total variation constraint $\sum_i |w_i-w_{i+1}| \le \ep J $, for $\ep = 0.01$ say.

A related idea appears in \cite{rubin2006method}, who show in simulations that power improves substantially by applying a smoothing spline to the weights, when the true effects $i \to \mu_i$ are smooth as a function of $i$. \cite{roeder2009genome} also find that smoothing the weights improves power when weights are estimated by sample splitting. 
\cite{ignatiadis2015data-driven} observe that adding a total-variation penalty or a penalty of the form $\sum_i|w_i-1| \le c$ can reduce variance. Princessp allows users to specify such constraints in a general principled framework.

\end{enumerate}

In specific applications, there can be many other constraints that incorporate problem-specific information and requirements. 

\subsection{Related work}
\label{rel_work}

There are many methods for multiple testing with prior information, partially reviewed by \cite{roeder2009genome} and \cite{gui2012weighted}, some mentioned earlier. Candidate studies---testing the top candidates based on prior information---are as old as statistics itself. They can be viewed as $p$-value weighting methods with binary weights. 

More general methods for multiple testing with prior information have been developed since at least the 1970's. Most work focuses on single-step Bonferroni procedures maximizing the average power and controlling the family-wise error rate. In seminal work, \cite{Spjotvoll1972optimality} described theoretically such procedures under very general conditions. 

Later work extended Spj\o tvoll's results in several ways. \cite{benjamini1997multiple} allowed weights in the importance of the hypotheses. \cite{roeder2009genome} and \cite{rubin2006method} found explicit optimal weights in the Gaussian model $\mathcal{N}(\mu_i,1)$. \cite{eskin2008increasing} and \cite{darnell2012incorporating} extended their framework to genome-wide association studies, accounting for correlations. 

\cite{westfall1998using} considered the model $\mathcal{N}(\mu_i,1)$ with proper priors on $\mu_i$. They gave small-scale algorithms to find optimal weights for the Bonferroni method via nonconvex optimization.  \cite{dobriban2015optimal} showed how to find the weights efficiently under a Gaussian prior. 

Less is known beyond the Bonferroni method.  \cite{ holm1979simple}'s step-down method can use weights, and \cite{westfall2001optimally} and \cite{westfall2004weighted} considered finding optimal weights for small $J$. \cite{genovese2006false} showed that the weighted Benjamini--Hochberg procedure controls the FDR and \cite{roquain2009optimal} showed how to choose weights optimally in a special asymptotic regime.  \cite{pena2011power} developed a general decision-theoretic framework for weighted family-wise error rate and FDR control.

\cite{bretz2009graphical} and follow-up work developed a graphical approach to stepwise multiple testing. To use this in large-scale applications, one needs to design a graph specifying how the levels are distributed upon rejection of each individual hypothesis.

Another line of work considers data re-use, by constructing weights using the same dataset where the tests are performed. \cite{rubin2006method} proposed a sample-splitting approach combined with smoothing. In a slightly different setup, \cite{storey2007optimal} developed the Optimal Discovery Procedure, maximizing the expected number of true discoveries, subject to a constraint on the expected number of false discoveries. 
\cite{sun2007oracle} developed a method maximizing the marginal False Nondiscovery Rate (mFNR) subject to controlling the marginal False Discovery Rate (mFDR), estimated the oracle procedure consistently in a hierarchical model. 

 \cite{bourgon2010independent} proposed an independent filtering approach where test statistics are independent of the prior information only under the null, not under the alternative hypothesis. Recently, \cite{ignatiadis2015data-driven} proposed the more general Independent Hypothesis Weighting framework. This promising approach focuses on the false discovery rate, relies on convex relaxations for efficient computations, and splits of the tests to ensure type I error control.

\section{Monotone $p$-value weights}
\label{mon}

The key practitioner constraints motivating Princessp were the need for monotonicity and boundedness of one-sided Gaussian weights. We now study that important example in detail. We develop algorithms to compute  bounded monotone weights for one-sided Gaussian tests, and explore their performance in simulations and GWAS data analysis. This is a helpful example because unconstrained Gaussian Spj\o tvoll weights are well understood and provide a solid background for comparison.

Assuming the prior means are sorted such that $0 \ge \mu_1 \ge \mu_2 \ge  \ldots \ge  \mu_J$, the bounded monotone weights problem with lower bound $l$ and upper bound $u$ is: 
\begin{align}
\label{monotone_weights}
  \max_{w\in [0,1/q]^J} \mbox{  } & \sum_{i=1}^{J} \Phi\left(\Phi^{-1}\left(qw_i\right)-\mu_i\right) \,\,\, \mbox{   s.t. }  \sum_{i=1}^{J} w_i = J.  \\
  \mbox{   s.t. }  &\,\, l \le w_1 \le w_2 \le  \ldots \le  w_J \le u. 
\end{align}

One justification for monotonicity is that larger effects are worth more. Another justification is that for a sufficiently small significance level $q$, the optimal Spj\o tvoll weights are in fact monotone. 

\begin{proposition}[Monotonicity of one-sided Spj\o tvoll weights for small $q$]
\label{mon_spjot}
Let us denote $M = \max_i \{\mu_i^2/2\}$ and $G(c) = \sum_{i=1}^{J}$ $\Phi(\mu_i/2 + c/\mu_i)$. Suppose that the significance level $q$ is small enough that $q \le G(M)/J$. Then the unconstrained Spj\o tvoll weights $w_i=w(\mu_i)$ defined by Eq. \eqref{Spjotvoll_weights} are monotone increasing in $|\mu_i|$.
\end{proposition}

\begin{proof}
As shown  by \cite{rubin2006method} and \cite{roeder2009genome}, the Spj\o tvoll weights are $w(\mu_i) = \Phi(\mu_i/2 + c/\mu_i)/q$, for the unique $c$ such that $G(c) = Jq$. Since all means $\mu_i <0$ are negative, $G(c)$ is decreasing. By the assumption on $q$, $Jq = G(c) \le G(M)$, hence $c \ge M$. 

Now,  the map $\mu \to l(\mu) = \mu/2 + c/\mu$ has derivative $l'(\mu) = 1/2 - c/\mu^2$, which is negative if $c\ge \mu^2/2$. By assumption on $M$, and by the above conclusion, $c \ge M \ge \mu_i^2/2$ for all $i$. Hence the map $l$ is decreasing in a range including all $\mu_i$. Therefore the weights $w_i = \Phi(l(\mu_i))/q$ are monotonically decreasing as a function of $\mu_i$, as desired. 
\end{proof}

A similar but somewhat more involved statement also holds for two-sided Spj\o tvoll weights. 
\begin{proposition}[Monotonicity conditions for two-sided Spj\o tvoll weights]
Under the assumptions of Proposition \ref{two_side_formula}, and with the notations used there, let $M = \max_i \{\mu_i^2/2\}$. Then the equation
$\log(a) - m = [2a/(a^2-1)^{1/2}-1]M$ has a unique solution $a^*$ on $(1,\infty)$. If $H[qa^*\exp(-m)] \ge J$, then the two-sided optimal weights $w_i = w(\mu_i;\lambda)$ defined in \eqref{two_sided_Spjotvoll_weights} are monotone increasing in $|\mu_i|$.
\end{proposition}

\begin{proof}
As shown in Proposition \ref{two_side_formula}, under our conditions, the weights are $w_i = w(\mu_i;\lambda)$, where  $w(\mu;\lambda) = 2\Phi\left(  - \arccosh\left[\lambda \exp(\mu^2/2)/q \right]/|\mu| \right)/q$, and  $\lambda\ge q\exp(-m)$ is the unique constant such that $\sum_i w_i(\lambda) = J$. Clearly, the weights are an even function of $\mu$. 

To show that the weights are monotone increasing in $\mu$, it is enough to ensure that the map $\mu \to l(\mu) = \arccosh\left[\lambda \exp(\mu^2/2)/q \right]/\mu$ is decreasing for $\mu \in (0,\infty)$. For this, let us re-parametrize by $x = \lambda \exp(\mu^2/2)/q$, and denote $c = q/\lambda$, so that $\mu = [2\log(cx)]^{1/2}$. It is enough to show that $l$ is decreasing as a function of $x$. Now $l(x) = \arccosh(x)/[2\log(cx)]^{1/2}$, so that 
$$l'(x) = \frac{(x^2-1)^{-1/2}[2\log(cx)]^{1/2}-\arccosh(x)\cdot 1/2 \cdot [2\log(cx)]^{-1/2} \cdot 2/x}{2\log(cx)}.$$
To ensure $l'(x) \le 0$, it suffices to argue that 
$$\frac{x}{(x^2-1)^{1/2}} [2\log(cx)]-\arccosh(x)\le 0,$$
or equivalently that $\cosh\left[2 x\cdot \log(cx)/(x^2-1)^{1/2}\right]\le x.$
Now, $\cosh(x) = [\exp(x)+\exp(-x)]/2 \le \exp(x)$ for $x\ge 0$, so, after taking logarithms, it is enough to show that  $2 x\cdot \log(cx)/(x^2-1)^{1/2} \le \log(x)$.
Suppose now that $x \ge a$ for some constant $a >1$. Since $x \to x/(x^2-1)^{1/2}$ is decreasing for $x>1$, it suffices to show that $b \log(cx) \le \log(x)$, where $b = 2a/(a^2-1)^{1/2}>1$. This is equivalent to $(cx)^b \le x$, or $x^{b-1} \le c^{-b}$, or also $x \le c^{-b/(b-1)}$. Recalling the definition of $x$, this inequality and $x\ge a$ are equivalent to 
$$a \le \frac{\lambda}{q} \exp(\mu^2/2) \le \left(\frac{\lambda}{q}\right)^{b/(b-1)}$$
or also $\log(a) - \log(\lambda/q) \le \mu^2/2 \le \log(\lambda/q)/(b-1)$. The same inequality can be written as
$$\log(\lambda/q) \ge \max( \log(a) - \mu^2/2, [2a/(a^2-1)^{1/2}-1]\mu^2/2).$$ 
Since this has to be true with $\mu = \mu_i$ for all $i$, it is enough to ensure
$$\log(\lambda/q) \ge \max( \log(a) - m, [2a/(a^2-1)^{1/2}-1]M).$$
Since $f_1(a) = \log(a) - m$ is a strictly increasing function of $a$ going from $-m$ to $\infty$ on $(1,\infty)$, while $f_2(a) = [2a/(a^2-1)^{1/2}-1]M$ is a strictly decreasing function of $a$ going from $+\infty$ to $M$ on $(1,\infty)$, there is a uniqe $a^*$ where $f_1(a)=f_2(a)$. Choosing this $a^*$, we obtain that the requirement is $\lambda \ge q a^* \exp(-m)$. 

Since $\lambda$ is defined by $H(\lambda) = J$, and $H$ is strictly decreasing, this is equivalent to $H[q a^* \exp(-m)] \ge J$. This finishes the proof. 
\end{proof}

While the above propositions hold under a mild constraint, the optimal weights are not bounded away from zero or infinity in general. This must be imposed as a separate constraint. 

We illustrate monotone weights, and their dependence on the bounds. We take $J = 10^3$ hypotheses, set the $p$-value cutoff $q = 10^{-7}$, and draw random effect sizes $\mu = -(|Z_1|,\ldots, |Z_J|)$, where $Z_i$ are iid standard normal.  First we vary $l$ from 0 to 0.75 in steps of size 0.25, keeping $u=\infty$. Next we vary $u$ from 1.25 to 2 in steps of size 0.25, keeping the $l=0$ (Fig. \ref{l_u_weights}). 

\begin{figure}
\centering
\begin{subfigure}{.33\textwidth}
  \centering
  \includegraphics[scale=0.25]{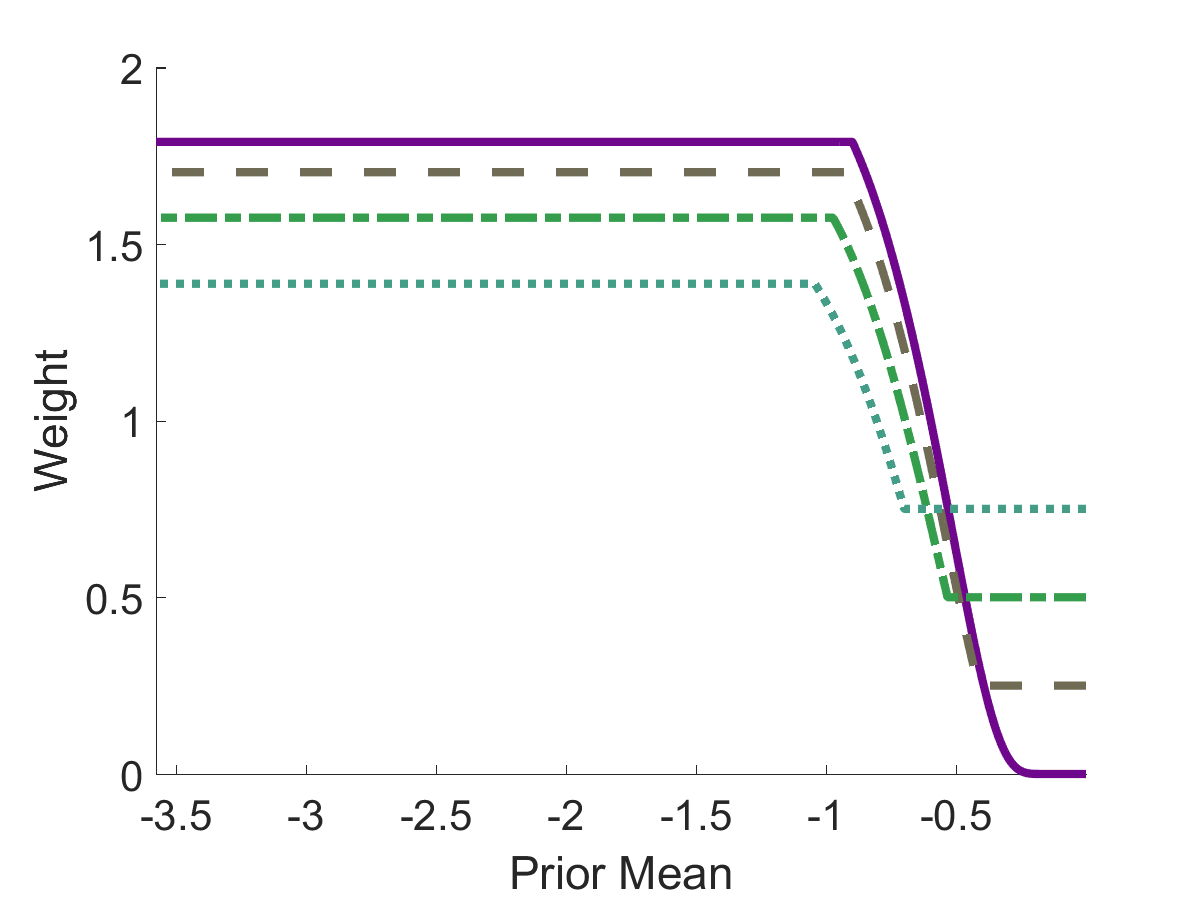}
  \caption{Lower bounded weights.}
\end{subfigure}%
\begin{subfigure}{.33\textwidth}
  \centering
  \includegraphics[scale=0.25]{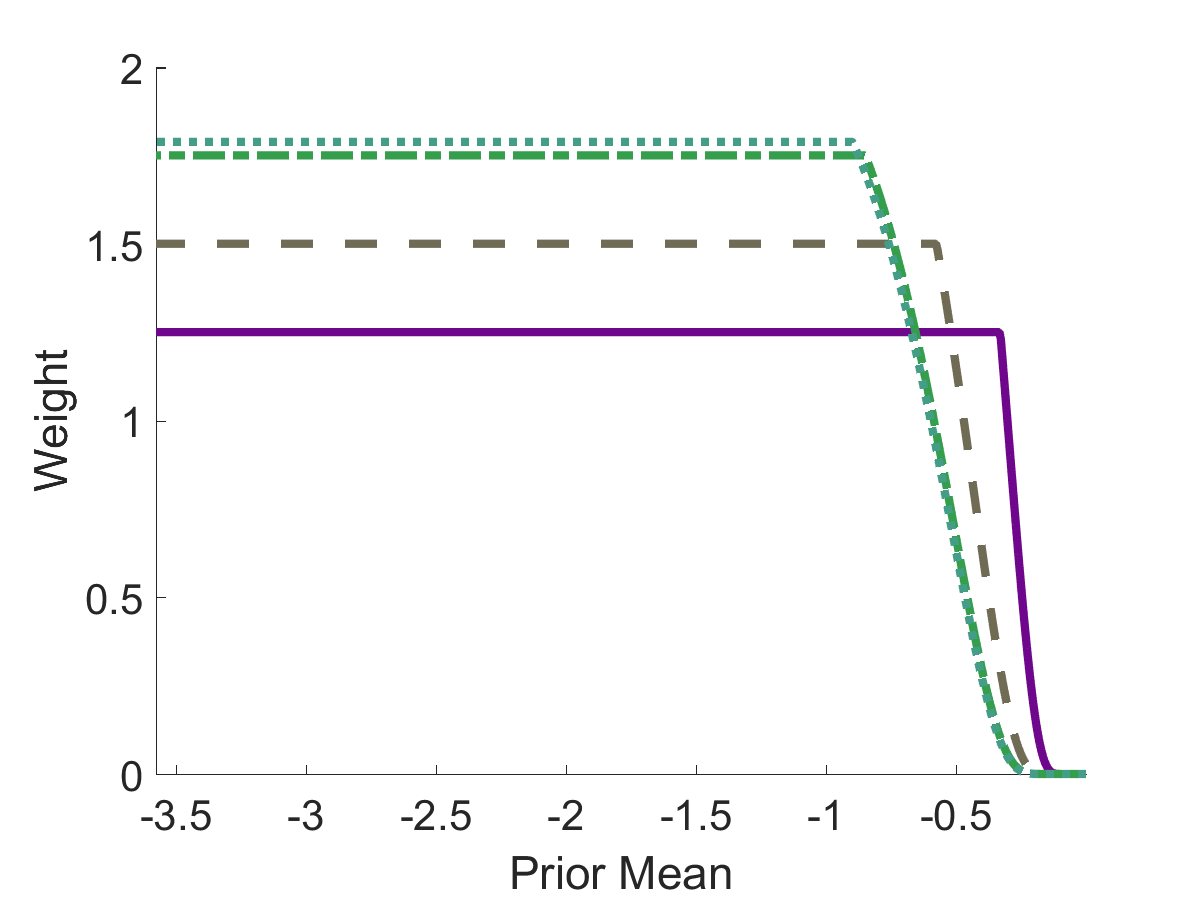}
  \caption{Upper bounded weights.}
\end{subfigure}
\begin{subfigure}{.33\textwidth}
  \centering
  \includegraphics[scale=0.25]{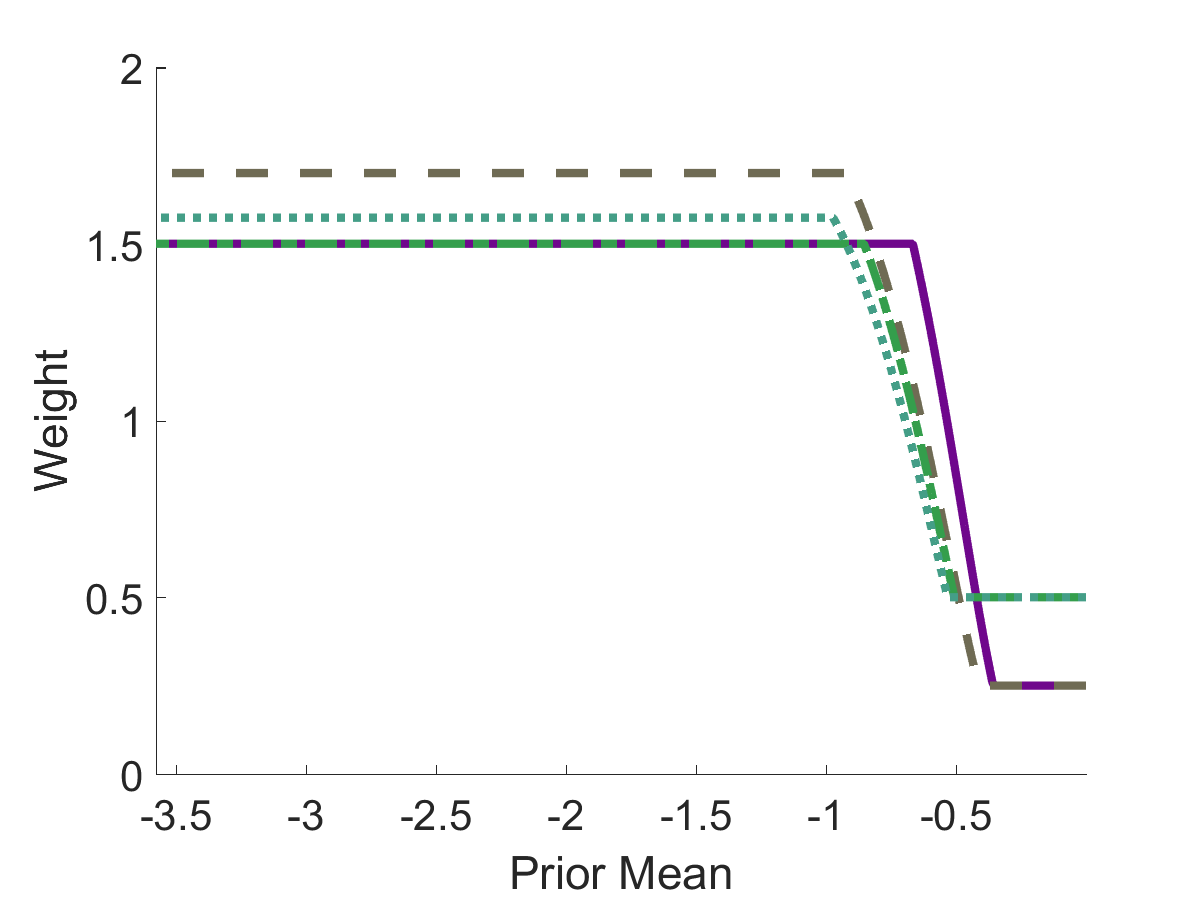}
  \caption{Bounded weights.}
\end{subfigure}
\caption{Upper and lower bounded weights.}
\label{l_u_weights}
\end{figure}

Lower bounded weights (Fig. \ref{l_u_weights}(a)) are flat near the two endpoints, and increase sharply in between. This may be surprising because flatness is not required in the optimization problem. In addition, increasing the lower bound also leads to a decreased upper bound, a property that is not immediately obvious theoretically.  Upper bounded weights (Fig. \ref{l_u_weights}(b))  have a similar but steeper shape. The weights are still flat, however, they are not automatically lower bounded (strictly above 0) anymore.   Finally, we vary both $u$ and $l$ (Fig. \ref{l_u_weights} (c)), with $l$ in $(0.25,0.5)$, and $u$ in $(1.5,1.75)$. The weights have the same general shape. 

\subsection{Power loss compared to Spj\o tvoll weights}
\label{sec:power_loss}

Since bounded weights are suboptimal to Spj\o tvoll weights when the model is correct, it is interesting to understand the loss in power. This may help to form heuristics for the choice of lower and upper bounds. On Fig. \ref{power_loss} we report the results of a simulation where  $J=10^4$, $q = 0.05/J$, the means are generated as iid negative absolute Gaussian, while the lower bounds are $(10^{-3}, 5\cdot10^{-3}, 10^{-2}, 5\cdot10^{-2}, 10^{-1}, 5\cdot10^{-1},0.9)$ and the upper bounds are $(2, 10, \infty)$. We define the power as the value of the maximized objective function in Eq. \eqref{monotone_weights}. We also display the power of the unweighted Bonferroni method (lowest horizontal line), and the optimal Spj\o tvoll weighted Bonferroni method (highest horizontal line). 

In this setting, an upper bound of two leads to a severe power loss of at least 50\%, regardless of the lower bound. However, an upper bound of 10 leads to only a small loss in power, while an upper bound of $\infty$ leads to nearly full power if the lower bound is at most $0.1$. We conclude that lower bounds below 0.1 and upper bounds above 10 seem to lead to a small loss in power. 

\begin{figure}
  \centering
  \includegraphics[scale=0.33]{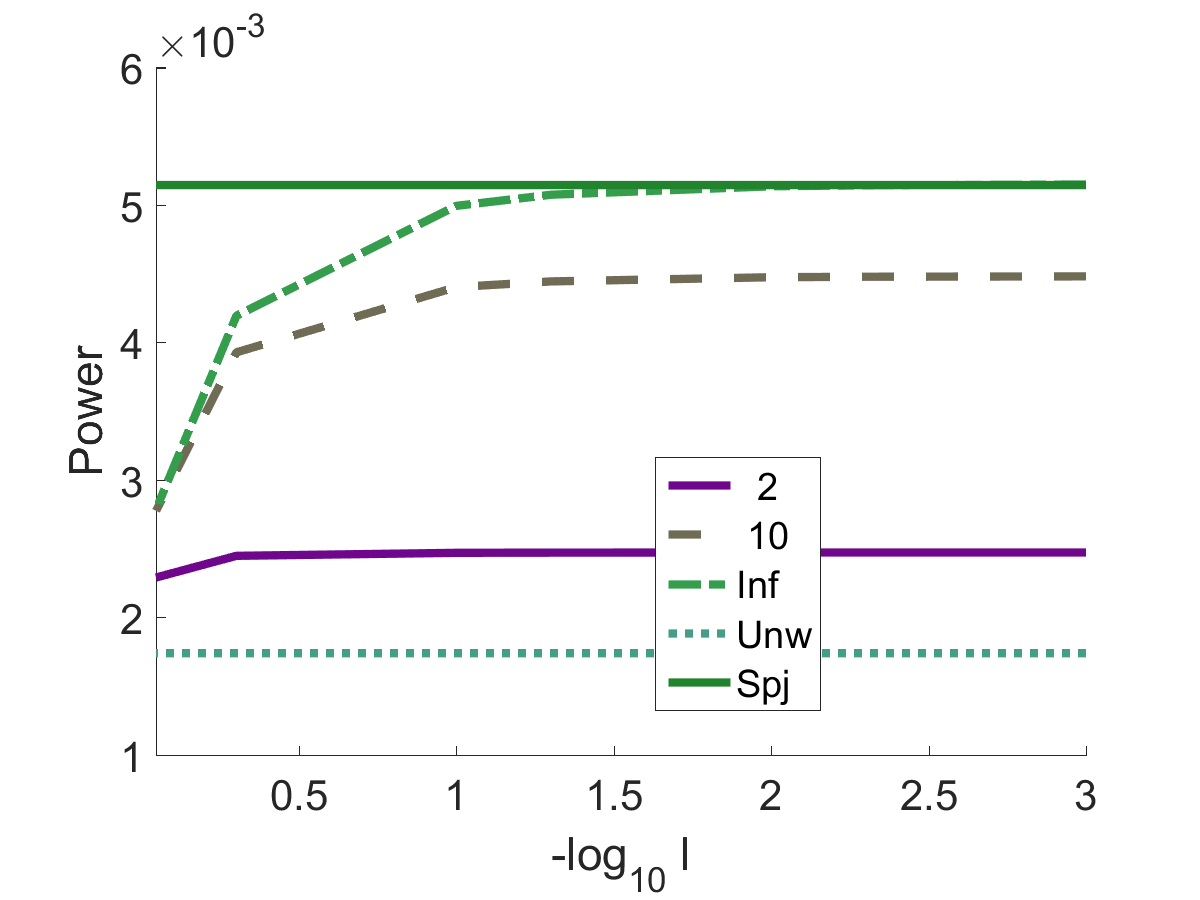}
  \caption{The loss in power of bounded weights compared to Spj\o tvoll weights.}
\label{power_loss} 
\end{figure}

\subsection{Data analysis example}
\label{sec:data}

To illustrate the empirical performance of Princessp, and specifically of bounded monotone weights, we analyze a standard set of datasets on genome-wide association studies (GWAS). We follow the protocol and methodology laid out in \cite{dobriban2015optimal}. There, we analyzed five studies on four complex human traits and diseases:  CARDIoGRAM and C4D for coronary artery disease \citep{schunkert2011large, coronary2011genome}, blood lipids  \citep{teslovich2010biological}, schizophrenia \citep{schizophrenia2011genome}, and estimated glomerular filtration rate (eGFR) creatinine \citep{kottgen2010new}.  

In addition to these, here we also include the 90Plus dataset \citep{deelen2014genome}, which compares the lifespan of a sample of Caucasian individuals living at least 90 years with matched controls (see the Supplement). These studies have $p$-values for the marginal association of 500,000--2.5 million single nucleotide polymorphisms (SNPs) to the outcome. More detail is provided in Sec. \ref{detailed_data}.

Testing SNPs one at a time using unweighted Bonferroni is typically the first analysis performed in genome-wide association studies \citep[see e.g.,][]{schunkert2011large, coronary2011genome,teslovich2010biological,schizophrenia2011genome,
kottgen2010new}. Therefore, in our data analysis we compare directly the state of the art unweighted Bonferroni method to weighted Bonferroni methods.

We analyze several pairs of these datasets, starting with those that were already included in \cite{dobriban2015optimal}. As a positive control for our method, we use CARDIoGRAM as prior information for C4D. Motivated by the Bayesian analysis of \cite{andreassen2013improved}, we use the blood lipids study as prior information for the schizophrenia study. Motivated by the comorbidity between heart disease and renal disease \citep{silverberg2004association}, we use the creatinine study as prior information for the C4D study. 

In addition to these pairs, we add four new examples: 

\begin{enumerate}

\item We switch the roles of the two heart disease studies, and use C4D as prior information for CARDIoGRAM. 
\item We use the 90Plus dataset as a target study, and check if studies on heart disease and schizophrenia can increase the number of hits. This is a challenging example, as the 90Plus data set has weak signal \citep{deelen2014genome}. We use the above three datasets as prior information as they seem to be the most promising from our prior work \citep{fortney2015genome}.
\end{enumerate}

\begin{table}[h]
\caption{Number of significant loci for five methods on three examples. Top: results pruned for linkage disequilibrium. Bottom: results without pruning. The score of each method is also reported. The weighting schemes compared are unweighted (Un); Spj\o tvoll (Spjot); filtering (Filter), Bayes, and Monotone (Mon) with $l=0.01,0.1,0.5$ and $u=\infty$. CG stands for CARDIoGRAM, SCZ for the schizophrenia study, and eGFRcrea for the creatinine study}
\begin{tabular}{rrr@{\hskip 0.3in}rrr@{\hskip 0.3in}rrr@{\hskip 0.3in}rrr}
&Un&Spjot&\multicolumn{3}{c}{Filter($-\log P$)}&\multicolumn{3}{c}{Bayes($\sigma$)}&\multicolumn{3}{c}{Mon($l$)}\\
Parameter&&&2&4&6&$0.1$ &1&10&$0.01$ &$0.1$&$0. 5$ \\
\\  
\multicolumn{9}{l}{Pruned}    \\  
CG $\to$ C4D &4  & 11   & 10  &   10   & 6& 10  &   8   & 4  &     11    &  10&9  \\
C4D $\to$ CG &9  & 14  & 22  & 11  &  3& 14    & 12  &   9  & 16   &  15 &  15  \\
Lipids  $\to$ SCZ    &  4  & 1  & 2  &2  &2&  1  & 1  & 5  &   2    &     2    &   4    \\
eGFRcrea $\to$ C4D    & 4   &  2   & 1 & 0  & 1& 2 & 4  & 4& 4 &   4      & 4   \\
CG $\to$ 90Plus    & 1   & 1 & 1& 0& 0 & 1& 1 & 1& 1 & 1 & 1 \\
C4D $\to$ 90Plus    & 1   & 1 & 1& 1& 0 & 1& 1 & 1& 1 & 1 & 1 \\
SCZ $\to$ 90Plus    & 1   & 0 & 0& 0 & 0& 0& 1 & 1& 1 & 1 & 1 \\
\\  
Score (Pruned)   & 0  &  -1 & -1 & -2 & -5&-1  & 1 &   1  &      1  &  1     &  2   \\     
\\
\multicolumn{9}{l}{Unpruned}  \\                                                                                                                             
CG$\to$ C4D  & 29   & 45 &  40   & 48 &34&  44   &  39 &29 &45 & 43&40  \\
C4D $\to$ CG &38  & 39  & 68 & 48 &30& 39 & 39  &38  &  53 & 52 & 48  \\
Lipids  $\to$ SCZ & 116 &214 & 217 & 96&39  & 214 & 223&123&217&220&225\\
eGFRcrea $\to$ C4D& 29 &18 &1 &0 &1 &18   &23 &  29 &23 & 25& 29     \\
CG $\to$ 90Plus    & 4   & 3 & 1& 0& 0& 4& 4 & 4& 4 & 4 & 4 \\
C4D $\to$ 90Plus    & 1   & 2 & 2& 1 & 0& 2& 1 & 1& 2 & 2 & 2 \\
SCZ $\to$ 90Plus    & 2   & 0 & 0& 0 & 0&0& 2 & 2& 2 & 2 & 2 \\  
\\  
Score (Unpruned)  & 0 &  1  & 1 & -2 &-6  &2  & 2 &1  & 3  &  3 &  4   \\   
\\
Total Score& 0  & 0 &  0 & -4 & -11  &1  & 3 & 2  &  4  & 4  &  6   \\         
\end{tabular}
\label{data_analysis_results}
\end{table}

For each pair of studies, we restrict to the SNPs that appear in both. For each SNP, we have a two-sided $p$-value $P_{0i}$ in the prior study, typically computed from a normal meta-analysis. From this we can back out the prior $z$-score $T_{0i}  = \Phi^{-1}(P_{0i}/2)$. We choose the sign so that $T_{0i}<0$, and optimize our weights for a a one-sided follow-up test in the same direction. This corresponds to a test for replicated sign. We do this because we have discussed monotone bounded weights for one-sided tests. We rescale the effect size $\mu_i = (N_i/N_{0i})^{1/2}T_{0i}$, as in \cite[][Sec. 6.2]{dobriban2015optimal} to estimate the effect size in the new study. Here $N_i$ and $N_{0i}$ are the current and prior sample size for the $i$-th variant, respectively.  We control the family-wise error rate at $\alpha=0.05$. 

We then compute optimal $p$-value weights with various methods. In \cite{dobriban2015optimal}, we reported the results of five weighting schemes: unweighted Bonferroni testing, as well as weighted Bonferroni testing using Spj\o tvoll weights, Bayes weights \citep{dobriban2015optimal}, filtering, and exponential weights \citep{roeder2006using}. Filtering selects all tests below a prior $p$-value cutoff, weighting them equally.  Here we report new results using Princessp, specifically bounded monotone weights with lower bounds $l=0.01, 0.1, 0.5$, and upper bound $u = \infty$. The results for upper bound $u= 10$ are very similar. We exclude exponential weights, as other methods performed better in our earlier work. 

As in \cite{dobriban2015optimal}, we prune the significant single nucleotide polymorphisms for linkage disequilibrium using the DistiLD database \citep{palleja2012distild}. We compute a score $s_{md}$ for each weighting method $m$ on each data set $d$. This is defined as $+1$ if the weighting method increases the number of hits compared to unweighted testing, 0 if it leaves it unchanged, and $-1$ otherwise. The score $s_{m}$ of a weighting method $m$ is the sum of scores over datasets. We compute scores separately for the pruned and unpruned analysis, and add them to find the final score. Table \ref{data_analysis_results} shows the number of significant loci. 

Princessp, i.e. monotone weighting, shows a good performance for all settings. In the pruned analysis, it has a score of two for the lower bound $l=0.5$. It increases the number of rejections in two out of the seven settings, and keeps it the same in the others.  In the unpruned analysis, it has a score of four with $l=0.5$, increasing the number of hits in four settings, and never decreasing it. Monotone weights with other lower bounds $l=0.1,0.01$ also lead to more hits in some cases, but they can also decrease the number of discoveries. 

Monotone weights perform well compared to the other methods. In particular, filtering decreases the number of hits in many cases, especially when the $p$-value cutoff is $10^{-4}$ or $10^{-6}$. Bayes weights, especially for $\sigma=1$, have a relatively stable performance. With $\sigma=1$, they decrease the number of hits in only two settings, while increasing it in five pairs. However, they decrease the number of hits more frequently for other $\sigma$ values. Moreover, monotone weights have a larger number of discoveries than Bayes weights. 

No weighting scheme increases the number of pruned SNPs in the 90Plus study. Even in this challenging example, however, monotone weights seem more stable than the others, as they do not decrease the number of hits.

\section{Optimization methods}
\label{sec:optim}

\subsection{General remarks}
\label{sec:optim_gen}

The Princessp optimization problems presented in Sec. \ref{cvx} are convex programs with convex inequality constraints. Here we propose using the log-barrier interior point method, which is a general approach to such problems \citep[e.g.,][Sec. 11]{boyd2004convex}. However, it is not immediate a priori that this approach will work well in our case. Indeed, we find that for monotone weights the straightforward application of the method leads to severely ill-conditioned KKT systems for large problems. Therefore, we develop a new subsampling method to avoid ill-conditioning. 

For concreteness, we will focus on Gaussian one-sided weights presented at the beginning of Sec. \ref{cvx}, but the more general case is similar; only the ROC function changes. To start, we discuss a few analytic properties of the optimization problem. The ROC function  $ f(w) = \Phi\left(\Phi^{-1}\left(qw\right)-\mu\right)   $
(plotted on Fig. \ref{fig:plot_obj})
is defined without ambiguity on $(0,1/q)$, and can be defined by continuity at the two endpoints. Indeed, as $w \to 0$, $\Phi^{-1}\left(qw\right) \to -\infty$, so that $f(w) \to 0$. Similarly $f(w)\to1$ at the other endpoint. 
Further, $f$ is concave on $[0,1/q]$, and importantly, its derivatives are convenient to compute. Indeed, denoting the $z$-score $z = \Phi^{-1}\left(qw\right)$, we have $  f'(w) = q \exp\left(\mu z - \mu^2/2\right) $
and
$  f''(w) = q^2 \mu \sqrt{2 \pi} \exp\left(z^2/2 + \mu z - \mu^2/2\right) $.
Thus $f$ is concave in $w$ for all $\mu \le 0$, and it is strictly concave for $\mu <0$. The second derivative is unbounded as $w \to 0$, because $\exp(z^2/2)$ is much larger than the other terms. 

\begin{figure}
  \centering
  \includegraphics[scale=0.35]{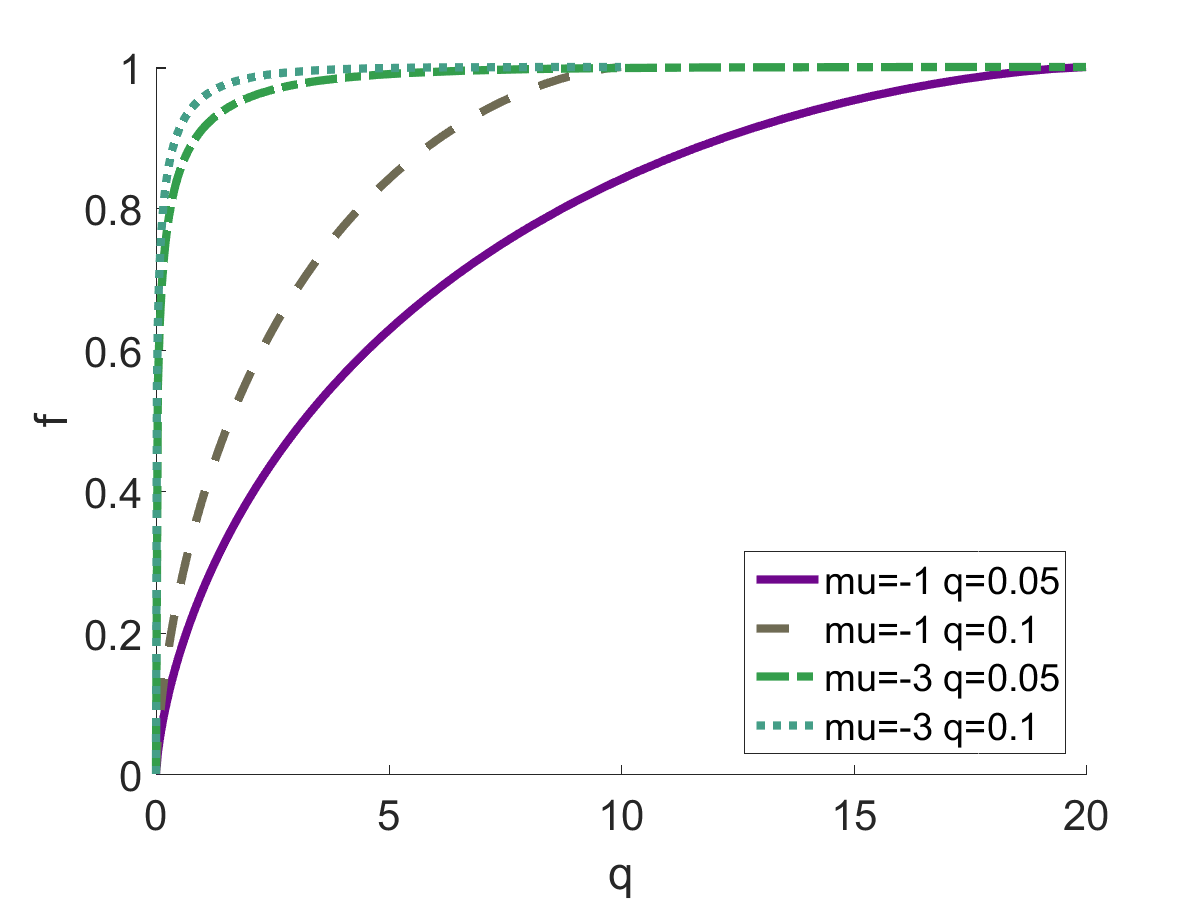}
\caption{Plot of the function $f$ appearing in the objective for various parameters $\mu,q$.}
\label{fig:plot_obj}
\end{figure}

The log-barrier method solves a sequence of equality-constrained problems indexed by $l=1,2,\ldots$, replacing the inequality constraints by a penalty $-(1/t_l)$ $\sum_k \log (-f_k(w))$ for an increasing sequence of $t_l$. The equality-constrained problems are solved using Newton's method. Under the assumptions stated at the beginning of Sec. \ref{cvx}, it follows from the general convergence analysis for the barrier method, see p. 577, Sec. 11.3.3, in \cite{boyd2004convex} that the barrier method converges for solving the convex constrained Spj\o tvoll weights optimization problem. Since we assumed that an optimal $w^*$ exists, strong convexity is ensured by compactness of the set $S$.

In the two-sided normal testing example, the log-barrier method converges when restricted to the region $w_i < 1/(2q)-\ep$ for $\ep>0$, where the problem is strongly convex.

\subsection{Monotone weights}
\label{sec:mon_w}

In this section we explain our method for computing monotone bounded weights (Eq. \eqref{monotone_weights}). To enable efficient computation for problems with tens of millions of weights, we need to exploit the tridiagonal structure of the Hessian and solve the KKT systems arising the in the Newton steps in $O(J)$-time directly. For this we need to implement several steps of the optimization method (Algorithm \ref{IP}). 

The method involves a number of optimization parameters, for which we use the default settings from \cite{boyd2004convex} Sec. 11. 
In addition, we need to choose a strictly feasible starting point $w^*_0$.  For this, we let $u = (1,2,\ldots, J)^\top$, and $v$ be the mean-centered version of $u$. Then we let 
\begin{equation}
\label{start_point}
w^*_0 = \textnormal{e} + \delta v 
\end{equation}
where $ \textnormal{e}  = (1,1,\ldots,1)^\top$ is the all ones vector and $\delta>0$ is a small constant such that the vector $w^*_0$ is strictly feasible. This is clearly possible if $0 \le l < 1 < u \le 1/q$.

\begin{algorithm}
\caption{Barrier Method for Monotone Weights}\label{IP}
\begin{algorithmic}[1]
\Procedure{Monotone Weights}{}
\State Given  $0 > \mu_1 \ge \mu_2 \ge  \ldots \ge  \mu_J$, $q$, $l$, $u$
\State $w^*_0 \gets \text{strictly feasible starting point from \eqref{start_point}}$
\State Set the index $k \gets 0$
\State Set optimization parameters $t >0, \mu >1, \varepsilon >0$
\Loop
\State \text{Centering step: Solve \eqref{centering_step} with parameter $t$, starting from $w^*_k$}
\State $w \gets w^*_k$
\Loop
\State Compute Newton step $\Delta w_{nt}$ by solving system \eqref{KKT_sys}
\State Compute Newton decrement $\lambda(w)$ from \eqref{Newton_dec}
\State \textbf{quit} loop if $\lambda^2/2 < \varepsilon$
\State Line search: Choose step size $s$ by backtracking line search (parameters $\alpha, \beta$)
\State Update $w \gets w + s \Delta w_{nt}$
\EndLoop
\State \textbf{end loop}
\State Update $w^*_{k+1} \gets w$, $k \gets k+1$
\State \textbf{quit} loop if $J/t < \varepsilon$
\State Increase $t \gets \mu t$
\EndLoop
\State \textbf{end loop}
\State \textbf{return}  $w^*_{\varepsilon} \gets w^*_{k+1}$
\EndProcedure
\end{algorithmic}
\end{algorithm}

\subsubsection{Centering problem}
From now on we flip the sign of the objective, so that we are minimizing a convex function. For a penalty parameter $t>0$, the logarithmic barrier problem---or centering step---is the convex program:
\begin{align}
\label{centering_step}
 &\min_{w\in [0,1/q]^J} \mbox{  } -t \sum_{i=1}^{J} \Phi\left(\Phi^{-1}\left(qw_i\right)-\mu_i\right) - \sum_{i=0}^{J} \log(w_{i+1} - w_i) \\ 
 & \mbox{s.t. }  \sum_{i=1}^{J} w_i = J.  
\end{align}
Here we define the constants $w_0 := l$, $w_{J+1} := u$ for brevity. 

We now show that the Karush-Kuhn-Tucker (KKT) matrix is a sum of a tridiagonal and a rank one matrix, so that the KKT system can be solved in $O(J)$ time. Let $g$ be the objective function,  $\nabla g$ its gradient, and $\nabla^2 g$ its Hessian. Let us also denote by $\Delta w_{nt}$ the Newton step for $w$, and by $\nu$ the scalar dual variable corresponding to the sum constraint  \citep[][p. 577]{boyd2004convex}. Then the KKT system for finding $(\Delta w_{nt}, \nu)$ at a particular $w$ (suppressed for ease of notation) is
\begin{equation}
\label{KKT_sys}
\begin{bmatrix} \nabla^2 g & \textnormal{e} \\ \textnormal{e}^\top & 0 \end{bmatrix}  \left[ \begin{array}{c} \Delta w_{nt} \\ \nu \end{array} \right]  = \left[ \begin{array}{c} - \nabla g \\ 0 \end{array} \right] .
\end{equation}
Assuming $\nabla^2 g$ is invertible, standard linear algebra shows that the solution has the form:
\begin{align*}
\label{KKT_sol}
 \nu =  & - \left( \textnormal{e}^\top  \left(\nabla^2 g\right)^{-1} \textnormal{e} \right)^{-1}  \left( \textnormal{e}^\top  \left(\nabla^2 g\right)^{-1} \nabla g \right)^{-1} \\ 
 \Delta w_{nt} = & -  \left(\nabla^2 g\right)^{-1}  \left( \nu \,  \textnormal{e} + \nabla g \right) .  
\end{align*}
The Newton decrement is defined as 
\begin{equation}
\label{Newton_dec} 
\lambda(w) = \left( \Delta w_{nt}^\top \left(\nabla^2 g\right)^{-1} \Delta w_{nt} \right)^{1/2}.
\end{equation}
To solve the Newton system, we must calculate $\left(\nabla^2 g\right)^{-1} \textnormal{e}$ and $\left(\nabla^2 g\right)^{-1}  \nabla g$. First, the components of the gradient are:
$$
\frac{\partial g}{\partial w_i} =  - q t \exp\left(\mu_i z_i - \frac{\mu_i^2}{2}\right) - \frac{1}{w_i - w_{i-1}} - \frac{1}{w_i - w_{i+1}}
$$
where we denoted the $z$-score $z_i = \Phi^{-1}\left(qw_i\right)$.
Next, the diagonal components of the Hessian are
$$
\frac{\partial^2 g}{\partial w_i^2} =  - q^2 t \mu_i \sqrt{2 \pi} \exp\left(\frac{z_i^2}{2} + \mu_i z_i - \frac{\mu_i^2}{2}\right)  + \frac{1}{(w_i - w_{i-1})^2} + \frac{1}{(w_i - w_{i+1})^2}
$$
and the off-diagonal terms vanish except those on the band just one-off the diagonal, which are  $
\partial^2 g/(\partial w_i \partial w_{i+1}) =  -  1/(w_i - w_{i+1})^2$.
Hence, the Hessian matrix $\nabla^2 g$ is tridiagonal. With the notation $h_i = 1/(w_i - w_{i+1})^2$, the Hessian has the form
$ \nabla^2 g = H  + D $
where $H$ is the tridiagonal matrix
$$
H = 
 \begin{pmatrix}
  h_0 + h_1 & -h_1 & \cdots & 0\\
  -h_1 & h_1+h_2 & \cdots & 0 \\
  \vdots  & \vdots  & \ddots & \vdots  \\
  0 & 0 & \cdots & h_J + h_{J+1}
 \end{pmatrix}
$$
 and $D$ is the diagonal matrix
$
D = - q^2t\sqrt{2 \pi}\cdot \textnormal{Diag}\left( \mu_i \exp\left(z_i^2/2 + \mu_i z_i - \mu_i^2/2\right) \right).
$ 
Finally $\nabla^2 g = H + D$ is invertible because it is the sum of two positive definite matrices. Indeed the diagonal terms in $D$ are strictly positive if $\mu_i <0$. Further $H$ is positive definite, as for a vector $x = (x_1,\ldots, x_J)$, the quadratic form $x^\top H x$ equals 
$x^\top H x = h_1 x_1^2 + \sum_{i=1}^{J-1} h_{i+1} (x_i - x_{i+1})^2 + h_{J+1} x_J^2.$ 
Now $h_i >0$ for all $i$, so $x^\top H x =0$ implies that all $x_i$ are equal to 0. This shows that $H$ is positive definite.
 
 Hence, we have shown that $\nabla^2 g$ is positive definite, and thus invertible. Therefore, the KKT system involves the solution of a tridiagonal-plus-rank 1 linear system, taking $O(J)$ time. This is implemented using a standard sparse tridiagonal linear solver in MATLAB, which is stable since $H$ is positive definite. 

\subsubsection{Subsampling}
\label{refinements}

Based on the earlier simulation results (e.g., Fig. \eqref{l_u_weights}), and on intuition from isotonic regression, we expect that the solution may have many equal terms, lying on the boundary of the feasible set. It is known that this can lead to an ill-conditioned KKT system \cite[e.g.,][Ch. 17]{nocedal2006numerical}. 

To deal with this problem, we propose a subsampling method. If $J$ is larger than a constant $L$, we subsample $L$ means evenly spaced among the indices $1,\ldots,J$. Here we choose $L=10^4$. To avoid ties, we then subsample the remaining means $\mu_j$ starting from $\mu_1$ and discard the remaining terms $\mu_2,\mu_3,\ldots$, until the first index $j$ where $\mu_j < \mu_1 - c_0$. We then repeat this starting from $j$, and so on. We choose the small constant $c_0=10^{-6}$. We use the barrier method to compute weights on the subsample, and then interpolate linearly to the remaining means.   

This new subsampling approach is crucial to enable the application of Princessp to large-scale problems with millions of hypotheses. We also show in Section \ref{subsample} that subsampling does not affect too much the accuracy of the weights on smaller problems. On larger problems, it usually avoids the ill-conditioned KKT systems encountered by the naive barrier method.

\subsection{Experiments}

In this section we report the results of several experiments with our optimization method. 

\subsubsection{Accuracy of subsampling}
\label{subsample}

In Section \ref{refinements} we introduced a subsampling method to avoid the ill-conditioning of the KKT system. Here we show that the method does not lose too much accuracy compared to the full barrier method.  We set $J = 2\cdot10^4$, $q = 5\cdot10^{-3}$, draw random $\mu = -(|Z_1|,\ldots, |Z_J|)$ where $Z_i$ are iid standard normal, and set the bounds $l=0$, $u = \infty$. We compute weights using the standard barrier method ($w_b$), and using the barrier method with subsampling ($w_s)$. We then compute the absolute error $|w_b-w_s|$ and relative error (defined pointwise as $|w_b(i)-w_s(i)|/|w_b(i)|$). The two weighting schemes (Fig. \ref{fig:sub})  agree within at least two significant digits in terms of absolute error for all means. The relative error can be as large as $10^{-1}$, but only when the absolute error is smaller than $10^{-6}$, so that this still translates to a good accuracy. 

\begin{figure}
\centering
\begin{subfigure}{.5\textwidth}
  \centering
  \includegraphics[scale=0.33]{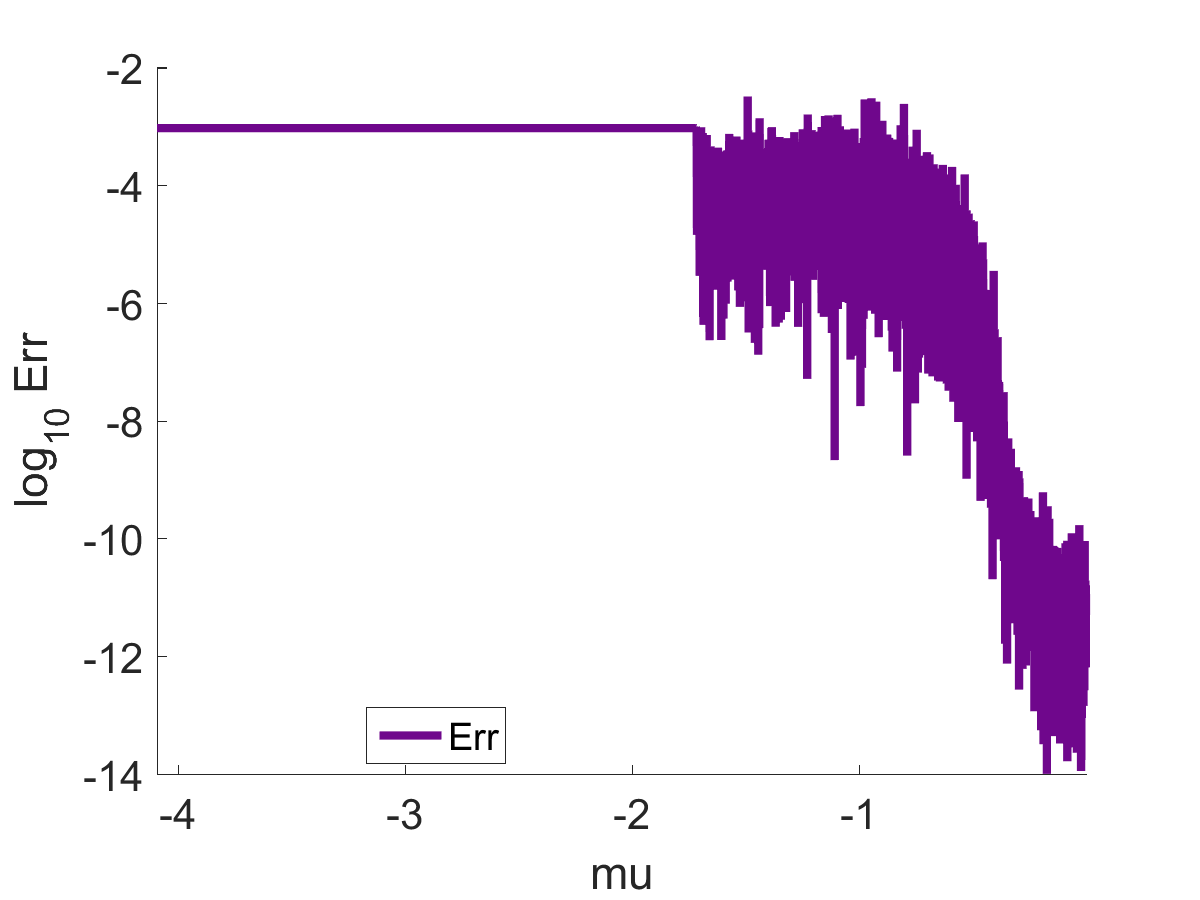}
  \caption{Absolute error.}
\end{subfigure}%
\begin{subfigure}{.5\textwidth}
  \centering
  \includegraphics[scale=0.33]{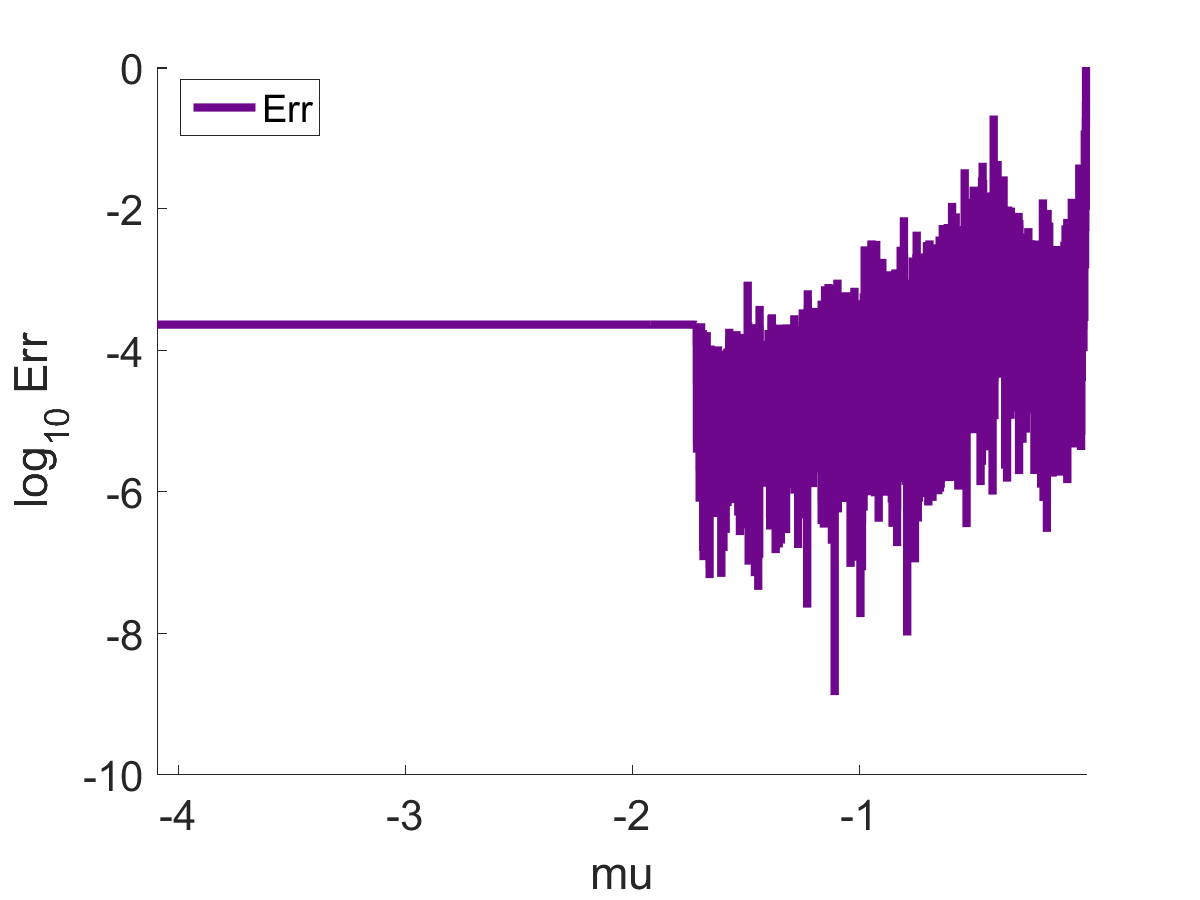}
  \caption{Relative error.}
\end{subfigure}
\caption{Accuracy of subsampling.}
\label{fig:sub}
\end{figure}
    
\subsubsection{Comparison with Spj\o tvoll}

We showed in Proposition \ref{mon_spjot} that Spj\o tvoll weights are monotone if $q$ is small. Thus we can compare the two methods on a problem where they should give the same results. We set $J = 10^3$, $q = 10^{-7}$, draw random $\mu = -(|Z_1|,\ldots, |Z_J|)$ where $Z_i$ are iid standard normal, and set $l=0$, $u = \infty$. Spj\o tvoll weights \eqref{Spjotvoll_weights} and monotone weights \eqref{monotone_weights} are visually indistinguishable (Fig. \ref{fig:spjot_mon}(a)).

\begin{figure}
\centering
\begin{subfigure}{.5\textwidth}
  \centering
  \includegraphics[scale=0.33]{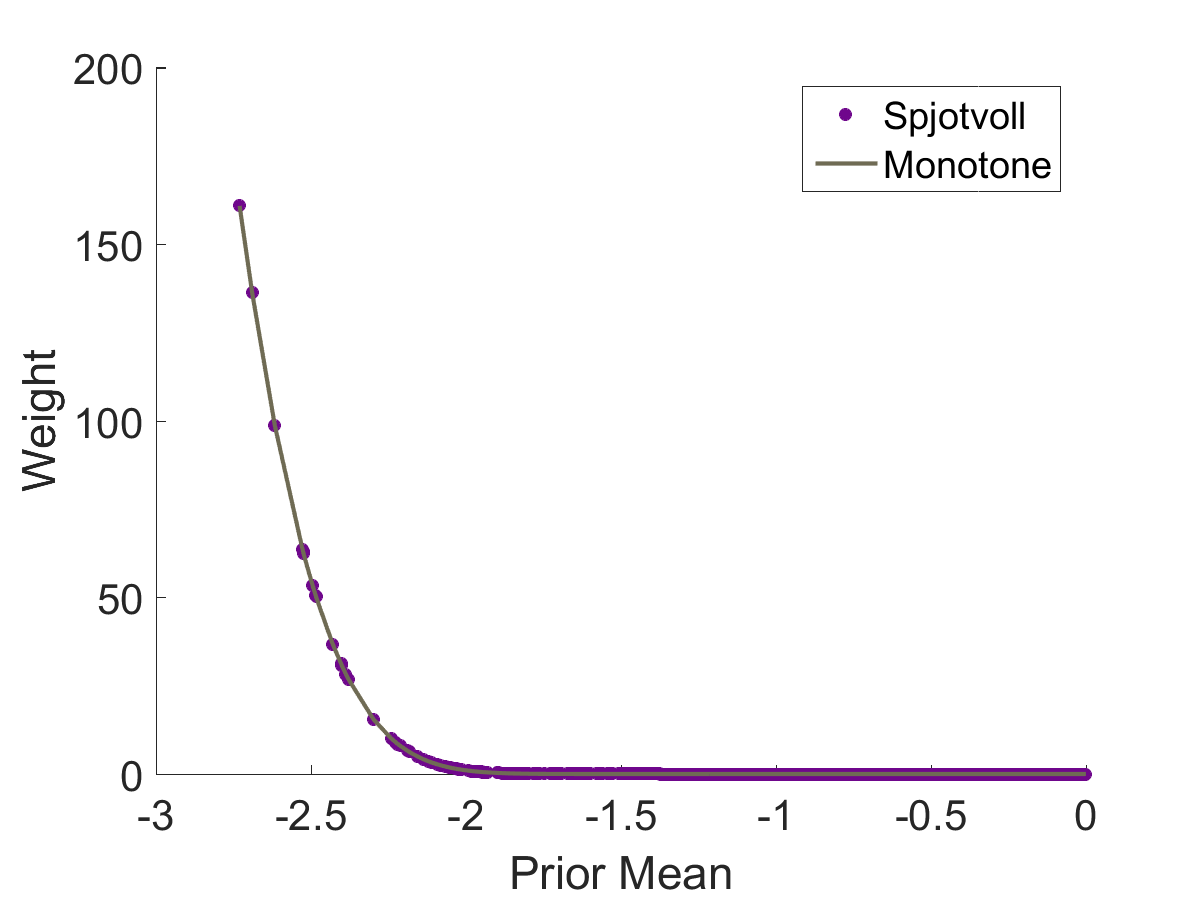}
  \caption{Spj\o tvoll and Monotone weights.}
\end{subfigure}%
\begin{subfigure}{.5\textwidth}
  \centering
  \includegraphics[scale=0.33]{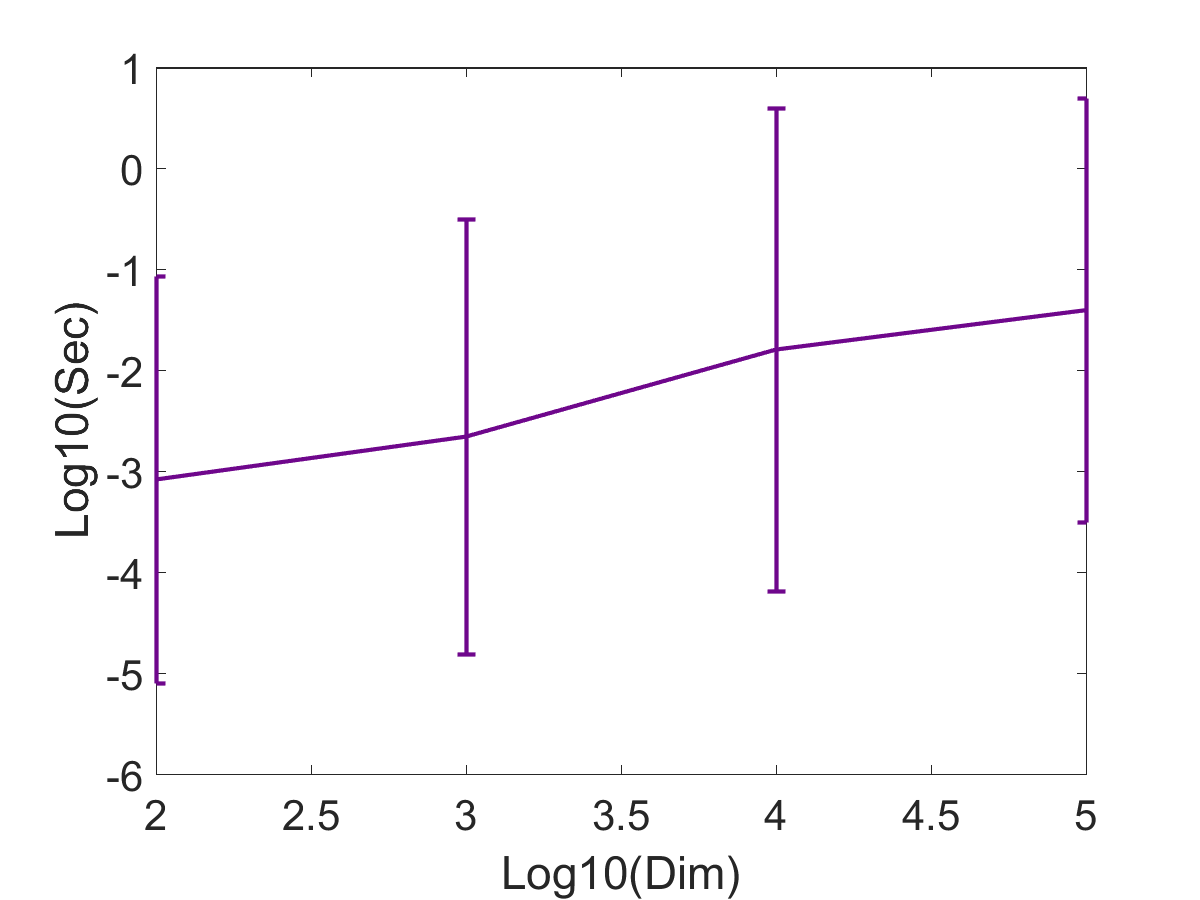}
  \caption{Timing.}
\end{subfigure}
\caption{(a) Spj\o tvoll and Monotone weights. (b) Base 10 log running time of barrier method. Averages and two standard errors over 50 independent MC trials.}
\label{fig:spjot_mon}
\end{figure}

\subsubsection{Running time}

To test the running time of our method, we vary $J$ in the range of $10^2$--$10^5$, choosing $\mu$ as above. The parameters $q$ and $l$ are chosen randomly as $q  = U/10$, $l = V/10$, where $U$, $V$ are independent uniform random variables on the unit interval. We also take $u=\infty$. The average running times over 50 simulations are shown in Fig. \ref{fig:spjot_mon}(b), and they are below a second even for the largest $J$.

\section{Discussion}

In this paper we developed the Princessp method, employing convex optimization for large-scale weighted Bonferroni multiple testing. Our approach enabled many different constraints. We found that bounded monotone weights perform well empirically in the analysis of genome-wide association studies.%

In particular, it appears that imposing a lower bound such as $w_i \ge 1/2$ can improve the empirical performance of $p$-value weighting methods. The reason is that the current $p$-values $P_i$ are multiplied by a small constant---at most two in the example above---and hence significant effects stay significant. Many of the current state of the art weighting methods, such as Spj\o tvoll, exponential, or Bayes weighting, do not have this property. Therefore, practitioners using them risk losing significant $p$-values. Bounded weights offer a principled way to avoid this problem. We think that this observation may go a long way in making $p$-value weighting methods more routinely applicable. 

\section*{Acknowledgements}

We thank Art Owen and Stuart Kim for discussion, and David Donoho for feedback on an earlier version of the manuscript. We acknowledge financial support from NSF grants DMS-1418362 and DMS-1407813, and an HHMI International Student Research Fellowship.

\section{Data sources}

\label{detailed_data}
\subsection{CARDIoGRAM and C4D}

CARDIoGRAM is a meta-analysis of 14 coronary artery disease genome-wide association studies, comprising 22,233 cases and 64,762 controls of European descent \citep{schunkert2011large}. The study includes 2.3 million single nucleotide polymorphisms. In each of the 14 studies and for each single nucleotide polymorphism, a logistic regression of coronary artery disease status was performed on the number of copies of one allele, along with suitable controlling covariates.
The resulting effect sizes were combined across studies using fixed effects or random effects meta-analysis with inverse variance weighting. Finally, two-sided normal $p$-values were computed. 

C4D is a meta-analysis of 5 heart disease genome-wide association studies, totalling 15,420 coronary artery disease cases and  and 15,062 controls \citep{coronary2011genome}. The samples did not overlap those from CARDIoGRAM. The analysis steps were similar to CARDIoGRAM.

The consortia require that the following acknowledgment be included:  Data on coronary artery disease / myocardial infarction have been contributed by CARDIoGRAMplusC4D investigators and have been downloaded from www.CARDIOGRAMPLUSC4D.ORG.

\subsection{Chronic Kidney Disease Consortium}

This is a genome-wide association study of kidney traits in
67,093 participants of European ancestry from 20 population-based cohorts \citep{kottgen2010new}. Estimated glomerular filtration rate creatinine was the trait with the largest sample size. There is no reported overlap with the samples from C4D. The analysis steps were similar to the previous two studies.

\subsection{Blood Lipids}
This is a genome-wide association study of blood lipids in a sample from European populations \citep{teslovich2010biological}. Triglyceride levels were one of the traits, with sample size 96,598, chosen here out of all lipids because of its previous appearance in \cite{andreassen2013improved}. Standard protocols for genome-wide association studies were used: linear regression analysis controlling for study-specific covariates, combined using fixed-effects meta-analysis. 
 
\subsection{Psychiatric Genomics Consortium - Schizophrenia}

This is a mega-analysis, which uses the raw data not just summaries of other studies, combining genome-wide association study data from 17 separate studies of schizophrenia, with a total of 9,394 cases and 12,462 controls \citep{schizophrenia2011genome}. They tested for association using logistic regression of schizophrenia status on the allelic dosages. 
 The overlap with the blood lipids study consists of 1,459 controls, which amounts to $12\%$ of controls in the schizophrenia study. The overlapping controls are from the British 1958 Birth Cohort of the Wellcome Trust Case Control Consortium. 
 
\subsection{90Plus Study - Aging}

\cite{deelen2014genome} performed a genome-wide association meta-analysis of 5406 long-lived individuals of European descent (aged at least 90 years). They combined the results of 14 studies originating from seven European countries. The analysis steps were similar to the ones above. This dataset was used in \cite{fortney2015genome}, and it is the most conveniently available aging data set among those analyzed in that paper.

{\small
\setlength{\bibsep}{0.2pt plus 0.3ex}
\bibliographystyle{plainnat-abbrev}
\bibliography{weighted}
}

\end{document}